\documentclass[reprint, amsmath, amssymb, showpacs, floatfix, longbibliography, prl, aps, twocolumn, superscriptaddress]{revtex4-1}
\usepackage[T1]{fontenc}
\usepackage{graphicx}
\usepackage{dcolumn}
\usepackage{amsmath,amssymb,amsfonts,bm,dsfont,units}
\usepackage{hyperref}
\hypersetup{colorlinks=true,citecolor=blue,linkcolor=blue, urlcolor=blue}
\hypersetup{linktocpage}
\usepackage{environ}
\usepackage{braket}
\usepackage{float}
\usepackage{latexsym}
\usepackage{multirow}
\usepackage[normalem]{ulem} 
\usepackage[caption=false]{subfig} 
\usepackage{ragged2e}
\usepackage{xcolor}
\usepackage{graphicx} 
\usepackage{enumitem}
\usepackage{tikz}
\usepackage{algorithm}
\usepackage{amsthm}

\definecolor{myblue}{RGB}{0,90,160}
\definecolor{myred}{RGB}{200,40,40}

\newtheorem{theorem}{Theorem}
\newtheorem*{theorem*}{Theorem}

\newtheorem{lemma}{Lemma}

\usepackage{color}
\definecolor{dkgreen}{rgb}{0,0.5,0}
\definecolor{midnightblue}{rgb}{0.39,0.58,0.93}

\definecolor{gold}{rgb}{0.85, 0.65, 0.13}

\newcommand{\change}{\color{black}}

\begin{document}

\title{Quantum metrology via partial quantum error correction}

\author{Yinan Chen}\thanks{Equal contribution}
\email{yinanc@caltech.edu}
\affiliation{Department of Physics and Institute for Quantum Information and Matter, California Institute of Technology, Pasadena, California, 91125, USA}

\author{Zongyuan Wang}\thanks{Equal contribution}
\email{zongyuan.wang15@outlook.com}
\affiliation{International Center for Quantum Materials, School of Physics, Peking University, Beijing 100871, China}

\author{Sisi Zhou}
\email{sisi.zhou26@gmail.com}
\affiliation{Perimeter Institute for Theoretical Physics, Waterloo, Ontario, N2L 2Y5, Canada}
\affiliation{Department of Physics and Astronomy, Department of Applied Mathematics, and Institute for Quantum Computing, University of Waterloo, Ontario N2L 3G1, Canada}

\date{\today}

\begin{abstract}
{\change We introduce a quantum metrology scheme in which partial quantum error correction suppresses local noise while preserving sensitivity beyond the standard quantum limit. The probes are superpositions of energetically distinct stabilizer sectors, whose sensing coherence is protected by measuring only a subset of local checks. We propose an adaptive weight strategy to balance two competing effects. Increasing the stabilizer generator weight suppresses parallel noise to $O(p^\delta)$, with $\delta$ growing with weight, but increases susceptibility to perpendicular noise. Our main theorem establishes a sufficient noise condition under which adapting this weight to system size preserves superlinear quantum Fisher information under simultaneous independent $X$ and $Z$ noise. We illustrate this strategy with repetition, Bacon–Shor, and toric codes, analyzing noise before and after phase imprinting using local checks and phase-imprinting operators without non-local connectivity.}
\end{abstract}
\maketitle

\noindent\textit{\textbf{Introduction.---}} Quantum metrology exploits many-body entanglement to estimate parameters beyond the standard quantum limit (SQL) \cite{Wasilewski2010,Wolfgramm2010,Aasi2013,Tse2019}. The achievable precision is quantified by the quantum Fisher information (QFI) through the Cramer--Rao bound, $\delta\theta\geq1/\sqrt{F_q[\rho_\theta]}$ \cite{toth2012,Hauke2016}. For $N$ unentangled probes, $F_q\sim N$, yielding the SQL $\delta\theta\sim N^{-1/2}$, whereas entangled probes can in principle reach the Heisenberg limit (HL), $F_q\sim N^2$ and $\delta\theta\sim N^{-1}$. We refer to the intermediate scaling $F_q\sim N^\nu$, with $1<\nu\leq2$, as super-SQL \footnote{In the presence of noise, the scaling depends on the quantum code, noise strength, and phase imprinter, as discussed in our example protocols.}.

However, good probe states, such as the Greenberger–Horne–Zeilinger (GHZ) state \cite{Giovannetti2006}, are susceptible to noise \cite{Ma2011,yinan2025,chen2026,PhysRevA.104.032407,LopezIncera2019allmacroscopic,Demkowicz-Dobrzański2012}. They can easily decohere, and thus, lose their ability to perform high precision estimations. Error-corrected quantum metrology addresses exactly this problem. Using techniques of quantum error-correction, we may preserve the entangled structure of the probe states before and after phase imprinting. Until now, the main methodology of error-corrected quantum metrology is to encode the GHZ state into the logical subspace of $M$ quantum codes as $\ket{\text{GHZ}}_L= \frac{1}{\sqrt{2}}\Big(\ket{0}^{\otimes M}_L + \ket{1}^{\otimes M}_L\Big)$ for $M \in \mathbb{Z}_{\geq 1}$ {\change with the phase imprinter acting as the logical Pauli-$Z$ operator $\sum_{i=1}^M Z_{L,i}$}~\cite{W2014,PhysRevLett.112.150802,ozeri2013}. To maintain the logical GHZ state, one requires error-correction using all checks of the code(s).

{\change
Here, we show that \textit{partial} quantum error correction can achieve super-SQL sensitivity by adapting the code and phase imprinter with the number of probes $N$. Our probes superpose energetically distinct stabilizer sectors, allowing a subset of local checks to suppress noise without erasing sensing coherence. Increasing the stabilizer generator weight raises the parallel-noise suppression order $\delta$ in $O(p^\delta)$ but increases susceptibility to perpendicular noise. Our \textit{adaptive weight strategy} balances these effects as $N$ grows. Under the code and decoding assumptions stated below, Theorem~1 proves that this strategy preserves superlinear quantum Fisher information under simultaneous independent $X$ and $Z$ noise before phase imprinting, provided $(1-2p_X)^2>\sqrt{4p_Z(1-p_Z)}$.

This Letter is organized as follows. We formulate the protocol using Calderbank--Shor--Steane (CSS) codes~\cite{steane_multiple-particle_1996,PhysRevA.54.1098} and present repetition, Bacon--Shor, and toric-code realizations with local checks and phase-imprinting operators, avoiding non-local connectivity. We then consider noise after phase imprinting and prove that, at an optimal operating point, the QFI under general i.i.d.\ Pauli noise is bounded below by that under the corresponding corrected dephasing channel (Theorem~2).
}

\begin{figure}
    \centering
    \includegraphics[width=1\linewidth]{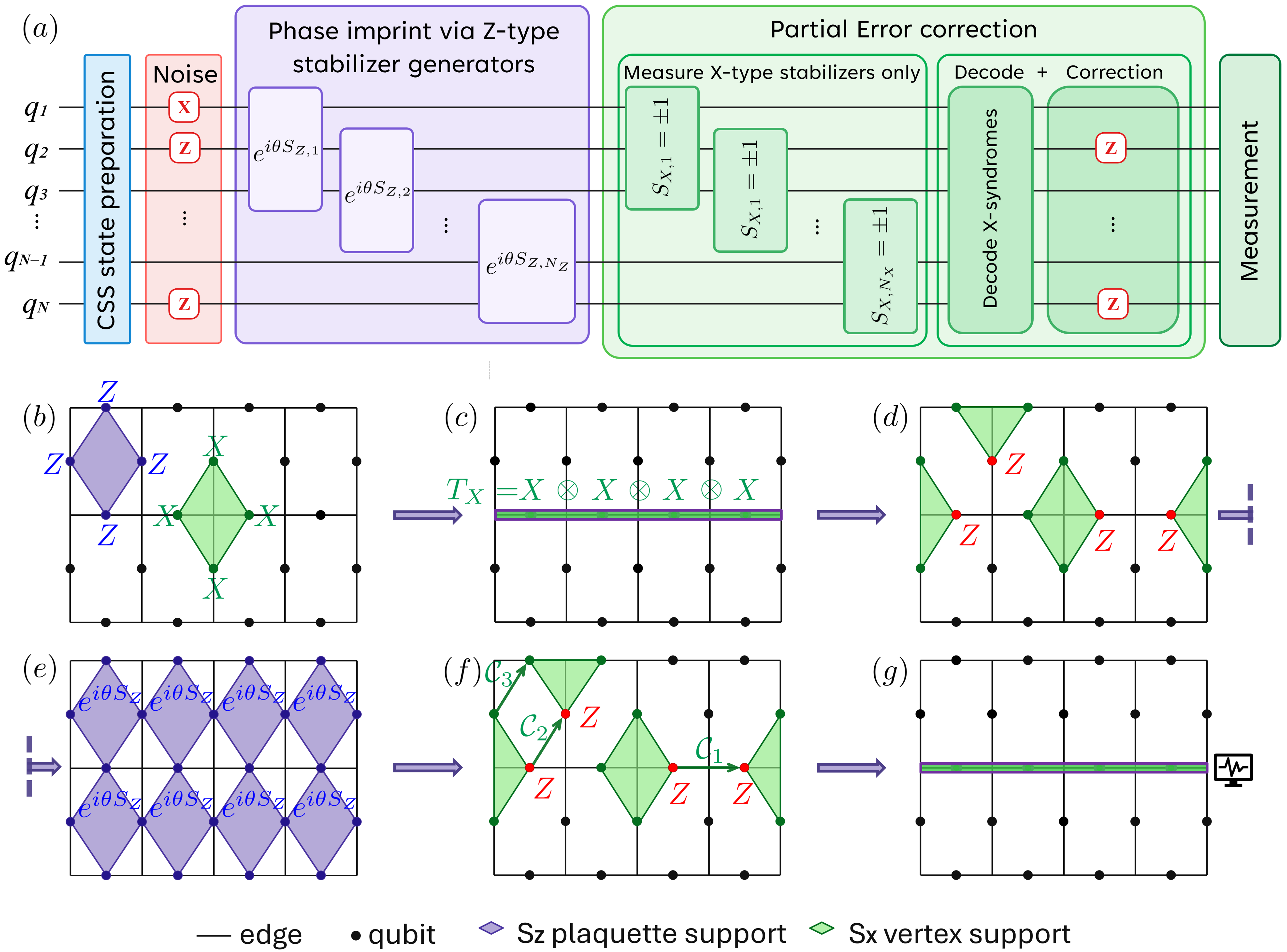}
    \caption{\label{fig:protocol2}
{\change Schematic of the partial-QEC metrology protocol and its toric-code realization. (a) After state preparation, noise acts before phase imprinting by $U(\theta)$. Partial QEC then measures only the $X$-type stabilizers, followed by syndrome decoding and correction, before the final measurement. (b--g) Toric-code example. (b) Qubits reside on edges; representative $Z$-type plaquette and $X$-type vertex stabilizers are shown in purple and green. (c) Measuring an $X$-type parity $T_X$ with odd overlap with every independent $S_Z$ prepares the two-branch probe of Eq.~\eqref{eq:CSS}, up to a known relative phase. (d) Pauli-$Z$ errors occur before phase imprinting. (e) The parameter is encoded by $U(\theta)=\prod_p e^{i\theta S_{Z,p}}$. (f) The $Z$ errors are diagnosed using only the $X$-type syndromes. The recovery $\mathcal C_1$ corrects the illustrated error, while alternative recoveries with the same syndrome, such as $\mathcal C_2$ and $\mathcal C_3$, can differ by an $S_Z$ stabilizer and hence flip the relative phase between the two branches. (g) The phase is read out from the $T_X$ parity by locally measuring its constituent single-qubit Pauli operators and multiplying their
outcomes.}}
\end{figure}

\noindent\textit{\textbf{The probe state.—}}
{\change The stabilizer generators of a CSS code can be chosen to be either $X$ type or $Z$ type, denoted by $S_{X,\alpha}$ and $S_{Z,\beta}$, respectively, where $\alpha,\beta$ label stabilizer generators and $i,j$ label physical qubits. Stabilizer generators are a set of commuting Pauli operators whose common eigenstates with eigenvalue $1$ form the code subspace, and we consider CSS-type encoding where $S_{X}$ (or $S_{Z}$) are Pauli operators consisting of only Pauli $X$ (or $Z$) operators. We construct the GHZ-like probe
\begin{equation}
    \ket{\mathrm{CSS}} = \frac{1}{\sqrt{2}} \left( \ket{S_Z=1\cdots 1} + \ket{S_Z=-1\cdots -1} \right),
\label{eq:CSS}
\end{equation}
where $\ket{S_Z=\pm1\cdots\pm1}$ denotes a simultaneous eigenstate of all $S_{Z,\beta}$ with eigenvalue $\pm1$, while both branches are implicitly taken to have eigenvalue $+1$ under every $S_{X,\alpha}$. Any remaining logical degrees of freedom of the original code may be chosen arbitrarily and consistently for the two branches~\footnote{The construction requires both stabilizer sectors in Eq.~\eqref{eq:CSS} to be nonempty. Relations among the $Z$-type stabilizers constrain the allowed eigenvalue assignments. If these constraints exclude the all-$-1$ sector, one may instead choose an allowed sector differing from it on only $O(1)$ stabilizer generators. This modifies the generator-eigenvalue separation by only $O(1)$ and therefore leaves the leading $N_Z$ scaling unchanged.}. One special case is the traditional GHZ state where $S_{Z}$ are local Pauli-Z operators acting on each qubit and $S_{X}$ is an empty set, in which case $\ket{\mathrm{CSS}} = \frac{1}{\sqrt{2}} (\ket{0^{\otimes N}} + \ket{1^{\otimes N}})$. Despite the GHZ-like notation, each branch is highly entangled. We fix independent generating sets of $X$- and $Z$-type stabilizers, denoted by $S_{X,\alpha}$ and $S_{Z,\beta}$, respectively. We denote by $N_Z$ the number of selected $Z$-type generators entering the phase generator and, for simplicity, by $l$ their common weight.}

{\change The probe can be prepared by measuring a Pauli-$X$ operator $T_X$ that has odd overlap with every $Z$-type stabilizer generator and hence anticommutes with each $S_{Z,\beta}$ [Fig.~\ref{fig:protocol2}(c)]. As an example, $T_X= X^{\otimes N}$ for the traditional GHZ state $\ket{\text{CSS}} = \frac{1}{\sqrt{2}} (\ket{0^{\otimes N}} + \ket{1^{\otimes N}})$. Applied to the ground state $\ket{S_Z=1\cdots1}$ of the original code, the two measurement outcomes produce the normalized states proportional to $(1\pm T_X)\ket{S_Z=1\cdots1}$. These are the two relative-phase choices of Eq.~\eqref{eq:CSS}, which differ only by a known $\pi$ phase shift; postselection is therefore unnecessary. This state preparation requires a nondemolition measurement of the global parity $T_X$, which can be implemented with an ancilla and local operations \cite{williamson_low-overhead_2026}; see also Remark~2 of the Supplemental Material (SM). The final readout is simpler: one can measure each single-qubit Pauli operator contained in $T_X$ locally and multiply the outcomes.}

Our phase imprinting is
\begin{equation}
    O_S:=\sum_{\beta=1}^{N_Z}S_{Z,\beta}, \qquad U(\theta)=e^{i\theta O_S} = \prod_{\beta=1}^{N_Z}e^{i\theta S_{Z,\beta}} .
    \label{eq:CSS-encoding}
\end{equation}
{\change Each factor acts only on the support of one $Z$-type stabilizer generator [Fig.~\ref{fig:protocol2}(e)]. Because the two branches of Eq.~\eqref{eq:CSS} have eigenvalues $\pm N_Z$ under $O_S$, the noiseless QFI is $F_q=4N_Z^2$, exhibiting Heisenberg scaling in $N_Z$. The corresponding stabilizer-level SQL is $F_q=4N_Z$, attained by a product state of $N_Z$ effective two-level probes. However, since each $S_{Z,\beta}$ generally acts on several physical qubits, the metrological advantage must ultimately be assessed in terms of physical resources. For $N$ physical qubits, the relevant benchmarks are the physical Heisenberg limit $F_q=\Theta(N^2)$, attained by an $N$-qubit GHZ state, and the physical SQL $F_q=\Theta(N)$, attained by an unentangled $N$-qubit probe.

The robustness of Eq.~\eqref{eq:CSS} arises from \textit{partial error correction}. The $Z$-type stabilizer generators generate the phase imprinting and are not measured for error correction, since measuring them would distinguish the two branches in Eq.~\eqref{eq:CSS}. Instead, we measure only the $X$-type stabilizers, which commute with both $S_{Z,\beta}$ and $T_X$ and therefore do not reveal the encoded phase. Following the GHZ terminology, we call Pauli-$Z$ errors \textit{parallel}, because they commute with the phase-imprinting operators, and Pauli-$X$ errors \textit{perpendicular}, because $X_j$ anticommutes with every $S_{Z,\beta}$ supported on qubit $j$. The measured $X$-type stabilizers can diagnose single-qubit parallel errors, whereas perpendicular errors are not corrected by this partial-QEC step.}

\noindent\textit{\textbf{Perpendicular and parallel noise.—}}
We consider independent i.i.d.\ $X$ and $Z$ errors acting on $\ket{\mathrm{CSS}}$. We focus on errors occurring before phase imprinting, which generally have the most detrimental effect \cite{Datta_2011,Dooley_2016}. {\change For the analytical result below, we assume that the total dephasing probability is asymptotically dominated by local failures associated with a weight-\(l\) \(Z\)-type stabilizer generator, each requiring errors on at least half of its support. Higher-order contractible failures and logical errors are taken to be asymptotically subleading. This assumption applies directly to the repetition code with majority-vote decoding and is consistent with the leading dephasing behavior of the toric code with MWPM decoding. The Bacon--Shor code, as a subsystem code, has a different row structure and will be treated separately.

Let $N_Z$ be the number of weight-$l$ $Z$-type stabilizer generators and assume
\begin{equation}\label{eq:check-incidence}
    N_Z=\frac{\gamma N}{l}+O(1),
\end{equation}
where $\gamma$ is independent of the number $N$ of physical qubits (For the toric code example, which we will see shortly, $\gamma=2$.). The $O(1)$ correction accounts for global stabilizer relations and does not affect the asymptotic scaling considered below. We consider local families in which each $Z$-type stabilizer generator overlaps with only $O(1)$ others, as in all examples below.}

{\change\begin{theorem}\label{thm:main}
Under the assumptions above, for simultaneous independent i.i.d.\ $X$ and $Z$ noise of strengths $0\leq p_X,p_Z<1/2$, there exists an increasing stabilizer generator weight $l(N)$ for which the QFI, $F_{q,XZ}$, remains super-SQL provided
\begin{equation}\label{eq:threshold}
    (1-2p_X)^2>\sqrt{4p_Z(1-p_Z)}.
\end{equation}
\end{theorem}

For equal noise strengths, $p_X=p_Z=p$, Eq.~\eqref{eq:threshold} gives the sufficient threshold
\begin{equation}
    p<p_\star\equiv
    \frac{1-\sqrt{(\sqrt5-1)/2}}{2}\simeq0.107.
\end{equation}
Pure perpendicular and parallel noise follow by setting $p_Z=0$ and $p_X=0$, respectively.

To prove Theorem~\ref{thm:main}, we first determine the effect of the two noise channels separately.

\begin{lemma}\label{cor:X}
Under i.i.d.\ $X$ noise
$\mathcal{E}_{X,j}(\rho)=(1-p_X)\rho+p_XX_j\rho X_j$, the QFI is
\begin{equation}\label{eq:Fqx}
    F_{q,X}=4(1-2p_X)^{2l}N_Z^2+\mathcal{O}(N_Z).
\end{equation}
\end{lemma}

\begin{proof}
Partial error correction acts trivially for $X$ noise. Since the channel strongly preserves $T_X$, the QFI is four times the variance of $O_S=\sum_{\beta=1}^{N_Z}S_{Z,\beta}$ \cite{frerot2024symmetry,chen2026}; see also the SM. In the dual picture, a Pauli product is multiplied by $1-2p_X$ for each qubit on which it anticommutes with the noise. Thus, separated pairs $S_{Z,\beta}S_{Z,\beta'}$ acquire the factor $(1-2p_X)^{2l}$. For the local code families considered here, only $\mathcal{O}(N_Z)$ diagonal or overlapping pairs are exceptional, giving Eq.~\eqref{eq:Fqx}.
\end{proof}

\begin{lemma}\label{cor:Z}
For corrected i.i.d.\ $Z$ noise in the majority-decoding setting above, let $P_{flip}$ be the probability that the residual error after recovery flips the $T_{X}$ parity. For $p_Z<1/2$,
\begin{equation}\label{eq:Fqz}
    P_{\rm flip}
    =\mathcal{O}\!\left[
        \frac{N_Z}{\sqrt l}
        \bigl(4p_Z(1-p_Z)\bigr)^{l/2}
      \right],
    \ F_{q,Z}=4(1-2P_{\rm flip})^2N_Z^2 .
\end{equation}
\end{lemma}

\begin{proof}
For one weight-$l$ phase-imprinting generator, the probability of
errors on at least $\lceil l/2\rceil$ of its qubits is
\begin{align}\label{eq:peff}
    q_l(p_Z)&= \sum_{m=\lceil l/2\rceil}^{l} \binom{l}{m}p_Z^m(1-p_Z)^{l-m}\nonumber \\ &=
    \mathcal{O}\!\left[\frac{[4p_Z(1-p_Z)]^{l/2}}{\sqrt l}\right].
\end{align}
Summing over the \(N_Z\) phase-imprinting generators gives a contribution \(\mathcal{O}[N_Zq_l(p_Z)]\). Higher-order contractible dephasing loops and logical errors are allowed, but we assume that their aggregate contribution does not change this leading scaling, so that \(P_{\rm flip}=O[N_Zq_l(p_Z)]\).
\end{proof}

\begin{proof}[Proof of Theorem~\ref{thm:main}] As shown in the SM, the two noise contributions factorize as $F_{q,XZ}=F_{q,X}F_{q,Z}/(4N_Z^2)$. Define $a_X=1-2p_X$ and $b_Z=\sqrt{4p_Z(1-p_Z)}$. For $0<p_Z<1/2$, Lemma~\ref{cor:Z} allows us to choose $N_Z(l)=\lfloor\eta\sqrt l\,b_Z^{-l}\rfloor$ where a sufficiently small constant $\eta>0$ ensures $P_{\rm flip}=\mathcal{O}(\eta)$ and hence keeps the $Z$-noise factor finite. Using $N=\frac{l}{\gamma} [N_Z+O(1)]$, we then obtain
\begin{equation}
    \frac{F_{q,XZ}}{4N} = \Omega\!\left[\frac{1}{\sqrt l}\left(\frac{a_X^2}{b_Z}\right)^l\right].
\end{equation}
This ratio diverges when $a_X^2>b_Z$, which is Eq.~\eqref{eq:threshold}. Since $N_Z(l)$ and $N$ both increase with $l$, this defines an increasing choice $l(N)$. The case $p_Z=0$ is immediate because there is no parallel-noise dephasing to suppress.
\end{proof}}

\noindent\textit{\textbf{Example protocols.—}}
We now discuss the GHZ, repetition, Bacon--Shor, and toric-code protocols for errors occurring before phase imprinting. Theorem~\ref{thm:main} applies directly to the repetition code with majority-vote decoding and is consistent with the leading dephasing behavior of the toric code with MWPM, while Bacon--Shor is treated separately due to its different subsystem-code geometry. 

\paragraph{GHZ state.}
As a warm-up, let $\ket{0}_L=\ket{0}^{\otimes N}$, $S_{Z,j}=Z_j$, and $T_X=\prod_jX_j$. Projection onto the $+1$ eigenspace of $T_X$ gives $\ket{\mathrm{GHZ}}=(\ket{0}^{\otimes N}+\ket{1}^{\otimes N})/\sqrt{2}$, with phase generator $O_S=\sum_j Z_j$. Under i.i.d.\ Pauli-$X$ and Pauli-$Z$ noise, respectively, the QFIs are
\begin{align}
F_{q,X}&=4(1-2p_X)^2N^2+16p_X(1-p_X)N,\\
F_{q,Z}&=4(1-2p_Z)^{2N}N^2 .
\end{align}
The first expression follows from the dual-channel calculation of Ref.~\cite{chen2026}, reproduced in the SM, whereas the second reflects direct dephasing of the GHZ coherence. This is the corresponding uncorrected $l=1$ case. Because there are no $X$-type checks, parallel Pauli-$Z$ errors cannot be diagnosed, and no adaptive increase of $l$ is available.

{\change\paragraph{Repetition code.}
The simplest nontrivial extension is the repetition-code protocol, whose suppression of parallel noise was previously studied in Ref.~\cite{W2014}. Here we consider simultaneous parallel and perpendicular noise and compare different choices of $l(N)$.

Partition the $N=Mn$ physical qubits into $M$ disjoint blocks of odd length $n$. In block $\mu$, define $S^X_{\mu,j}=X_{\mu,j}X_{\mu,j+1}$ and $S^Z_\mu=\prod_{j=1}^{n}Z_{\mu,j}$, with $O_S=\sum_\mu S^Z_\mu$ and $T_X=\prod_\mu X_{\mu,1}$. Thus, $l=n$ and $\gamma=1$.

Measuring the $S^X_{\mu,j}$ checks and applying majority decoding corrects every Pauli-$Z$ pattern of weight below $(n+1)/2$. This is exactly the majority-decoding case covered by Theorem~\ref{thm:main}. Defining the dephasing probability as $q_n(p)=\sum_{k=(n+1)/2}^{n}\binom{n}{k}p^k(1-p)^{n-k}$, the QFIs are
\begin{align}
    F_{q,Z}^{\rm rep}
     &=4M^2[1-2q_n(p_Z)]^{2M},\\
    F_{q,X}^{\rm rep}
     &=4(1-2p_X)^{2n}M^2
     +4[1-(1-2p_X)^{2n}]M .
\end{align}

Fig.~\ref{fig:p_star} compares two $n(N)$ strategies for equal noise strengths, $p_X=p_Z=p$. In the adaptive strategy, shown by the green circles in Fig.~\ref{fig:p_star}(a), we start from $n=3$ and increase $n$ by two whenever the QFI under simultaneous parallel and perpendicular noise returns to the physical SQL. The corresponding crossings are illustrated in Fig.~\ref{fig:p_star}(b). Alternatively, optimizing $n$ directly at each $(p,N)$ gives the blue curve in Fig.~\ref{fig:p_star}(a) and the blue staircase in the inset. Its intersection with the SQL exponent $\nu=1$ yields $p_\star^{\rm rep}\simeq0.107$, in agreement with Theorem~\ref{thm:main}.

For comparison, the green dotted staircase in the inset shows the larger $n(N)$ obtained by optimizing only the parallel-noise QFI, as in Ref.~\cite{W2014}. Robustness against simultaneous noise instead requires a more gradual increase of the stabilizer generator weight.}

\paragraph{Bacon--Shor code.} Theorem~\ref{thm:main} can be extended to subsystem CSS codes. Consider the Bacon--Shor subsystem code \cite{Bacon2006,Aliferis2007} on a periodic $m\times n$ lattice, with $m$ even, $n$ odd, and $N=mn$. Its gauge generators are $g^Z_{i,j}=Z_{i,j}Z_{i+1,j}$ and $g^X_{i,j}=X_{i,j}X_{i,j+1}$. We fix the $X$-gauge sector and choose a product of row-GHZ states (in the $X$ basis $\ket{\pm}$) as the probe. Projection onto the $+1$ eigenspace of $T_X=\prod_i X_{2i-1,1}$ then produces a state of the form in Eq.~\eqref{eq:CSS}.

We take $S^Z_i=\prod_{j=1}^{n}Z_{i,j}Z_{i+1,j}$ as the $Z$-type stabilizer generators used for phase imprinting, with $O_S=\sum_{i=1}^{m-1}S^Z_i$. Each such stabilizer has weight $l=2n$, and each qubit belongs to two of them, so $\gamma=2$. Measuring the $g^X_{i,j}$ gauges reduces Pauli-$Z$ correction to majority decoding within each row. Errors on even rows commute with $T_X$ and do not change the relative phase between the two branches, so a harmful decoding failure requires at least $(n+1)/2$ errors on an odd row. Since the full stabilizer has weight $l=2n$, this is not the half-of-$l$ condition assumed in Theorem~\ref{thm:main}; the dephasing probability is the row-majority tail $q_n(p)$ in the repetition code. However, with a different threshold value $p_\star$, the super-SQL condition in Theorem~\ref{thm:main} can still be met.

Using the same $q_n(p)$ as above, the QFIs are
{\change
\begin{align}
F_{q,Z}^{\rm BS}
 &=4(m-1)^2[1-2q_n(p_Z)]^{m},\\
F_{q,X}^{\rm BS}
 &=4(1-2p_X)^{4n}(m-1)^2 \nonumber\\
 &\quad
 +4\!\left[
 1+2(1-2p_X)^{2n}-3(1-2p_X)^{4n}
 \right](m-1) \nonumber\\
 &\quad
 +4\left[
 2(1-2p_X)^{4n}-2(1-2p_X)^{2n}
 \right].
\end{align}

The corresponding results are shown alongside the repetition code in Fig.~\ref{fig:p_star}. The adaptive rule $n\mapsto n+2$ gives the red squares in Fig.~\ref{fig:p_star}(a) and the red dashed curve in Fig.~\ref{fig:p_star}(b). Direct optimization gives the orange dashed curve and staircase in Fig.~\ref{fig:p_star}(a) and its inset. For large $n$, $q_n(p)=\mathcal{O}\!\left(n^{-1/2}[4p(1-p)]^{n/2}\right)$. Balancing this row-majority suppression against the perpendicular-noise factor gives the Bacon--Shor condition
change\begin{equation}
(1-2p)^4>\sqrt{4p(1-p)},
\end{equation}
change whose boundary is $p_\star^{\rm BS}\simeq0.074$. The optimized exponent in Fig.~\ref{fig:p_star}(a) crosses $\nu=1$ at the same value.}

\begin{figure}[t]
    \centering
    \includegraphics[width=0.87\linewidth]
    {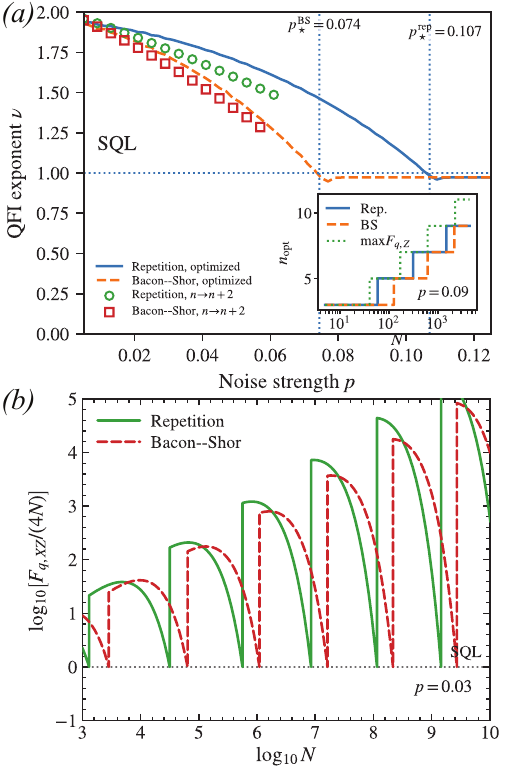}
    \vspace{-4mm}
    \caption{\label{fig:p_star}
    {\change Repetition and Bacon--Shor protocols under simultaneous parallel
    and perpendicular noise.
    (a) QFI scaling exponent $\nu$ for $p_X=p_Z=p$.
    Blue solid and orange dashed curves show the asymptotic exponents
    optimized over the repetition and Bacon--Shor code sizes; the SQL
    plateau begins at $p_\star^{\rm rep}\simeq0.107$ and
    $p_\star^{\rm BS}\simeq0.074$. Green circles and red squares show the corresponding adaptive protocols; values slightly below $\nu=1$ are finite-size effects, which are expected to converge to the SQL exponent as $N\rightarrow\infty$. The inset shows the selected code sizes at $p=0.09$; the green dotted staircase maximizes $F_{q,Z}$ under parallel noise alone.
    (b) Adaptive protocols at $p=0.03$, with the code size $n$ increased by two whenever the QFI returns to the SQL.}}
\end{figure}

\begin{figure}
    \centering
    \includegraphics[width=1\linewidth]{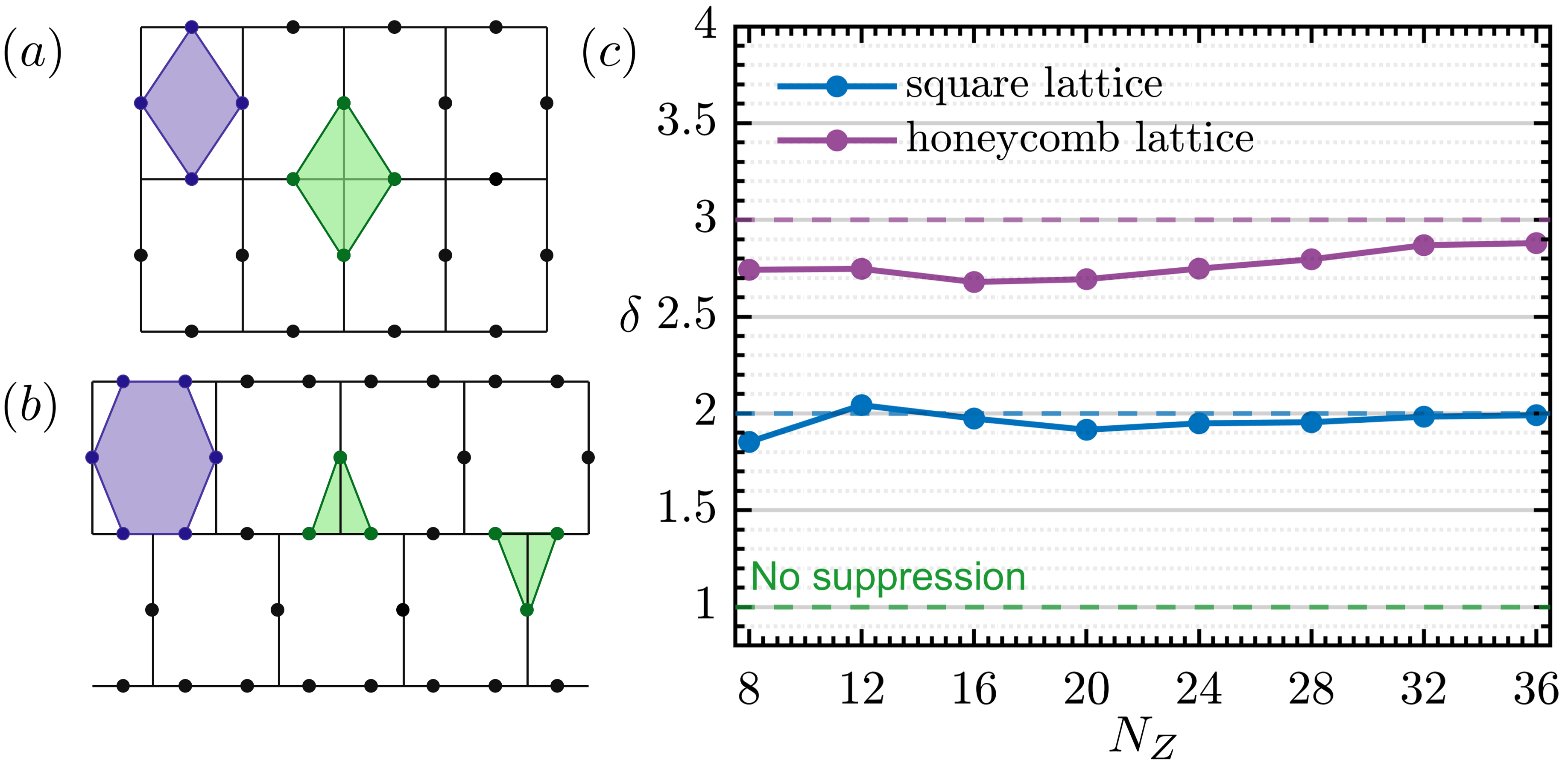}
\caption{\label{fig:toric_lattice}
{\change Increasing the toric-code plaquette weight by changing the underlying lattice. (a) Square lattice with uniform plaquette weight $l=4$. (b) Honeycomb lattice with $l=6$. Purple and green regions denote representative $Z$-plaquette and $X$-vertex stabilizers, respectively. (c) Fitted effective-noise exponent $\delta$ versus $N_Z$ for the two lattices. The results approach $\delta=\lceil l/2\rceil=2$ and $3$, respectively, demonstrating stronger suppression of parallel noise at larger $l$. The green dashed line, $\delta=1$, indicates the absence of error suppression.
}}
    \label{fig:subdivision}
\end{figure}

\paragraph{Toric code.}
Consider a toric code \cite{KITAEV20032,bravyi1998} with qubits on edges and stabilizer generators $S_{X,v}=\prod_{e\ni v}X_e$ and $S_{Z,p}=\prod_{e\in\partial p}Z_e$. We fix the eigenvalues of two commuting logical operators so that the two branches in Eq.~\eqref{eq:CSS} are unambiguously specified, and choose the plaquette stabilizers for phase imprinting, $O_S=\sum_p S_{Z,p}$. An $X$-type readout $T_X$ has odd overlap with each independent $S_Z$, as illustrated in Fig.~\ref{fig:protocol2}(c,d). Parallel Pauli-$Z$ errors are diagnosed by measuring the vertex stabilizers and applying minimum-weight perfect-matching (MWPM) decoding, Fig.~\ref{fig:protocol2}(f).

For a canonical stabilizer generator $S_{Z,p}$, a single-plaquette dephasing event requires at least $\lceil l/2\rceil$ physical errors. Thus, $p_{\rm eff}(l,p_Z)=O(p_Z^\delta)$ with $\delta=\lceil l/2\rceil$. This motivates the effective form
\begin{equation}
    F_{q,Z}\simeq 4[1-2p_{\rm eff}(l,p_Z)]^{2N_Z}N_Z^2 ,
\end{equation}
which captures the leading fixed-weight dephasing behavior observed numerically. In particular, the square lattice [Fig.~\ref{fig:subdivision}(a)] has $l=4$ and $p_{\rm eff}=O(p_Z^2)$, while the honeycomb lattice [Fig.~\ref{fig:subdivision}(b)] has $l=6$ and $p_{\rm eff}=O(p_Z^3)$ [Fig.~\ref{fig:subdivision}(c)], consistent with the small-$p$ suppression $p\rightarrow p^\delta$, with $\delta=\lceil l/2\rceil$, entering Lemma~\ref{cor:Z}.

{\change Alternatively, arbitrary even weights can be generated by subdividing the horizontal edges of the square lattice. Inserting $r$ vertices on each horizontal edge gives $l=4+2r$. This is only a refinement of the toric-code cellulation: contracting the added degree-two vertices recovers the original lattice, so the topology and number of logical qubits are unchanged. The original vertices support weight-four $X$ stabilizer generators, whereas the inserted vertices support weight-two $X$ stabilizer generators. Since every edge belongs to exactly two plaquettes, $lN_Z=2N$ and $\gamma=2$ throughout the family. Taking $r=O(\log N)$ therefore realizes the increasing $l(N)$ required by Theorem~\ref{thm:main} and suppresses parallel noise to progressively higher order.} The detailed calculation of the QFI under perpendicular noise is left to the SM.

\noindent\textit{\textbf{{Noise~after~phase~imprinting---}}} We also consider errors occurring after phase imprinting. The evolved probe is $\ket{\text{CSS}_{\theta}}=\frac{1}{\sqrt{2}}\left[e^{iN_Z\theta}\ket{S_Z=1\cdots1}+e^{-iN_Z\theta}\ket{S_Z=-1\cdots-1}\right]$, which is no longer an eigenstate of $T_X$. Nevertheless, the $S_X$ stabilizers still commute with the phase imprinting, so they can be measured without erasing the encoded phase. Perpendicular noise is especially benign in this setting. If the dual noise channel leaves the readout parity unchanged, $\prod_jE_j^\dagger(T_X)=T_X$, then $\text{Tr}[\prod_jE_j(\rho_\theta)T_X]=\text{Tr}[\rho_\theta T_X]=\cos(2N_Z\theta)$. A measurement of $T_X$ therefore retains classical Fisher information $4N_Z^2$ at the optimal operating point. Since a parameter-independent noise channel cannot increase the QFI above its noiseless value, the QFI remains exactly $4N_Z^2$. For parallel noise, the situation is the same as before because the phase imprinter commutes with the error channel. More generally, we can lower bound the QFI for arbitrary i.i.d.~Pauli errors:
{\change\begin{theorem}
At the operating point $\theta_\star=\frac{\pi}{4N_Z} \left(\mathrm{mod}\ \frac{\pi}{2N_Z}\right)$, the QFI of the probe state $\rho_\theta$ under i.i.d.\ Pauli errors described by $E_P[\rho_\theta]=(1-p_X-p_Y-p_Z)\rho_\theta+p_XX\rho_\theta X+p_YY\rho_\theta Y+p_ZZ\rho_\theta Z$ followed by partial error correction is lower bounded by the QFI under the i.i.d.\ dephasing channel $\mathcal{E}_{z}[\rho_\theta]=(1-p)\rho_\theta+pZ\rho_\theta Z$ followed by partial error correction, with $p=p_Y+p_Z$.
\end{theorem}}
We prove this theorem in the SM. {\change It shows that perpendicular errors after the phase imprinter do not reduce the achievable QFI below the corresponding dephasing-only value.} Moreover, since the dephasing channel commutes with the phase imprinter, all previous discussion on $Z$ noise before phase imprinting also applies here. As before, the sensing performance can be improved by increasing the stabilizer generator weight to suppress the parallel noise.

\noindent\textit{\textbf{Conclusion---}} In this Letter, we introduced a new scheme for error-corrected quantum metrology built on the idea of \textit{partial} quantum error correction. We analyzed the tradeoff of noise suppression. For parallel noise with strength $p$, the leading error is suppressed to a code-dependent higher order $p^\delta$ with $\delta>1$. When both parallel and perpendicular noises are present, we showed how to maintain QFI beyond the SQL by adapting the code and stabilizer generator weight as the system size increases. 

We thank Sara Murciano and Pablo Sala for insightful discussions. Z.W. is supported by the National Natural Science Foundation of China (Grant No.~12474491) and the Central University Fundamental Research Funds (Peking University). S.Z. acknowledges support from National Research Council of Canada (Grant No.~AQC-217-1) and Perimeter Institute for Theoretical Physics, a research institute supported in part by the Government of Canada through the Department of Innovation, Science and Economic Development Canada and by the Province of Ontario through the Ministry of Colleges and Universities.

\bibliography{reference}

\newpage
\onecolumngrid
\begin{center}
\Large{\bf Supplemental Materials}
\end{center}
\setcounter{section}{0}
\setcounter{theorem}{0}
\setcounter{page}{1}
\renewcommand{\thepage}{S\arabic{page}}
\setcounter{figure}{0}
\renewcommand{\thefigure}{S\arabic{figure}}
\setcounter{table}{0}
\renewcommand{\thetable}{S\arabic{table}}
\setcounter{equation}{0}
\renewcommand{\theequation}{S\arabic{equation}}

\section{Remarks}\label{Sec: Supplemental Materials}

\noindent\textit{\textbf{Remark 1.}} In this work, we introduced a sensing method that utilizes only partial QEC on the probe states constructed from CSS codes to achieve enough noise suppression to retain the probe states' super-SQL sensing performance. The main results are summarized as follows.

\begin{itemize}
    \item \textbf{Analytical expression for QFI, $F_{q,X}$, under perpendicular noise.} \\
    Building on the symmetry argument in Ref.~\cite{chen2026}, we analytically computed the QFI under perpendicular noise that (strongly) preserves the symmetry of the probe state, with an explicit and exact dependence on the stabilizer generator weight $l$. Although $F_{q,X}$ exhibits super-SQL scaling with respect to the number of stabilizer generators for any fixed $l$, it can fall below the SQL when expressed in terms of the number of physical qubits if $l$ grows too rapidly with system size. This means that $l$ cannot increase too rapidly to suppress parallel noise.
    \item \textbf{Improvement for QFI, $F_{q,Z}$, under parallel noise.} \\
    Increasing the stabilizer generator weight can strongly suppress parallel noise. For the repetition and toric-code families, a local dephasing event associated with a weight-$l$ stabilizer generator requires at least $\lceil l/2\rceil$ errors on its support, giving a leading power that grows with $l$. Bacon--Shor has a different row-majority condition, but its suppression order likewise increases with the strip width. These examples show that stronger parallel-noise suppression can be obtained by increasing the stabilizer generator weight, at the cost of greater sensitivity to perpendicular noise.
    \item \textbf{An adaptive method to maintain super-SQL performance under both noise types.} \\
    For the repetition and Bacon--Shor protocols, we compare direct optimization of the code size with an adaptive rule that increases the row/block size only when needed. Under equal parallel and perpendicular noise, the optimized scaling crosses the SQL at $p_\star^{\rm rep}\simeq0.107$ and $p_\star^{\rm BS}\simeq0.074$, respectively, as shown in Fig.~\ref{fig:p_star}. Furthermore, we prove that the QFI under simultaneous independent $X$ and $Z$ noise factorizes as $F_{q,XZ}=\frac{F_{q,X}F_{q,Z}}{4N_Z^2}$, which makes the trade-off between the two noise channels explicit.
    \item \textbf{Noise suppression after phase imprinting} \\
    We also considered the case where errors occur after phase imprinting, and showed that increasing the stabilizer generator weight $l$ suppresses general i.i.d.\ Pauli errors. Although we focused on probe states constructed from CSS codes, our approach applies more broadly to other metrologically useful states. For example, one could encode the Dicke states \cite{PhysRevLett.107.080504,Apellaniz_2015}. Furthermore, since a logarithmic increase of $l$ with $N$
is sufficient for robustness (see Fig.~\ref{fig:p_star}), the phase imprinter can be viewed as a sum of \emph{local} operators, making it potentially feasible on platforms that support both QEC and local parameter encoding, such as superconducting-qubit architectures \cite{acharya_quantum_2025,lacroix_scaling_2025} and neutral atom platforms \cite{bluvstein_fault-tolerant_2026,reichardt2025faulttolerantquantumcomputationneutral}. This suggests that local codes developed for quantum computation can also be advantageous in metrological settings.
\end{itemize}

\noindent\textit{\textbf{Remark 2.}} The parity operator $T_X$ generally has support across the system, but its two uses in the protocol should be distinguished. For the final readout, one simply measures the single-qubit Pauli operators appearing in $T_X$ and multiplies their outcomes, so the readout itself is local. State preparation instead requires a nondemolition measurement of the global parity. Such a measurement can be implemented with ancilla-assisted local gadgets; one possible fault-tolerant construction is the gauging of logical operators in Ref.~\cite{williamson_low-overhead_2026}. A full circuit-level analysis of this preparation step is beyond the scope of this work.

\noindent\textit{\textbf{Remark 3.}} We name a few more points that are beyond the scope of our work. Since our purpose is to introduce the partial QEC scheme, we have limited the analysis of our protocols in the simplest settings. First, we only considered errors on the physical qubits. A full error analysis that includes other types of circuit-level noise, such as noise from imperfect measurement, faulty gate implementation during syndrome extraction, or qubit loss (as seen, e.g. in Refs.~\cite{chen2026,sahu2026achievingheisenberglimitusing}) is left to future works. Second, we have not considered the most general decoherence model beyond a composition of i.i.d. $X$ and $Z$ noise before the phase imprinting and a general Pauli noise after the phase imprinting. Third, our discrete treatment of noise via quantum channels may miss features of the continuously applied signal and noise \cite{Chaves2013}. Apart from these, it would also be interesting to see how our current examples generalize to multi-parameter estimation tasks, where the commutation relations of the imprinters play a crucial role in determining the Fisher information matrix \cite{Ragy2016,Demkowicz2020}. It would also be interesting to see how taking multiple copies of our encoded state would increase robustness against noise as seen in e.g. Ref.~\cite{W2014,Trenyi_2024}. Such a discussion should be accompanied by an analysis on the tradeoff between robustness against noise and QEC resource costs.

\section{QFI under perpendicular noise}

In this section, we provide a method to compute the QFI under perpendicular noise for different codes used in the main text. We first sketch the proof of a general statement: the QFI equals $4$ times the variance of the generator of the imprinter (as in the pure state case), given that the symmetry is strongly preserved. The proof can also be found in \cite{chen2026} and a closely related version in \cite{frerot2024symmetry}.

Consider a unitary phase imprinter $U(\theta)=e^{i\theta O}$ and a strong symmetry of the probe state $T_{X}\rho=t_x\rho$ that \textit{anticommutes} with the generator $O$ $\left\{T_X,O\right\}=0$. Plugging into the error propagation formula, we have $\braket{T_{X}}_\theta=\text{Tr}\left[\ U^\dagger(\theta) T_{X} U(\theta)\rho\right]
=\text{Tr}\left[\ U^{\dagger2}(\theta) T_{X}\rho\right]=t_x\text{Tr}\left[\ U^{\dagger2}(\theta)\rho\right]$. Similarly, $\braket{T_X^2}_{\theta}=\text{Tr}\left[\ U^\dagger(\theta) T_{X}^2 U(\theta)\rho\right]
=\text{Tr}\left[\ T_{X}^2\rho\right]=t_x^2$. Taylor expanding $U^\dagger(\theta)$, we obtain $\braket{T_{X}}_{\theta}
= t_x \left[1-2\theta^2\braket{O^2} + \mathcal{O}(\theta^3)\right]$, yielding $\partial_\theta\braket{T_X}_\theta=t_x\left[-4\theta\braket{O^2}+\mathcal{O}(\theta^2)\right]$. Plugging back into the error propagation formula, we have:
\begin{equation}
    \delta\theta^{-2}
    =\frac{\left|\partial_\theta\braket{T_X}_\theta\right|^2}
    {\braket{T_X^2}_\theta-\braket{T_X}_\theta^2}
    =\frac{t_x^2\left[-4\theta\braket{O^2}
    +\mathcal{O}(\theta^2)\right]^2}
    {t_x^2-t_x^2\left[1-2\theta^2\braket{O^2}
    +\mathcal{O}(\theta^3)\right]^2}
    =4\braket{O^2}+\mathcal{O}(\theta^2).
    \label{eq:QFI_is_variance}
\end{equation}
Meanwhile, $T_X$ also serves as the optimal measurement, in the
sense that the QFI is always upper-bounded by four times the
variance of $O$, $F_Q\leq 4\,\mathrm{Var}(O)$. Consequently, the propagated error saturates the CRB for $\theta\rightarrow0$.

We finally remark that we do not assume that $T_X$ is Hermitian or that $t_x$ is real. For a general unitary transform $T_X$, we refer interested readers to \cite{chen2026} for a construction of a POVM using Hadamard test, whose classical Fisher information (CFI) saturates the QFI.

\begin{figure}
    \centering
    \includegraphics[width=0.6\linewidth]{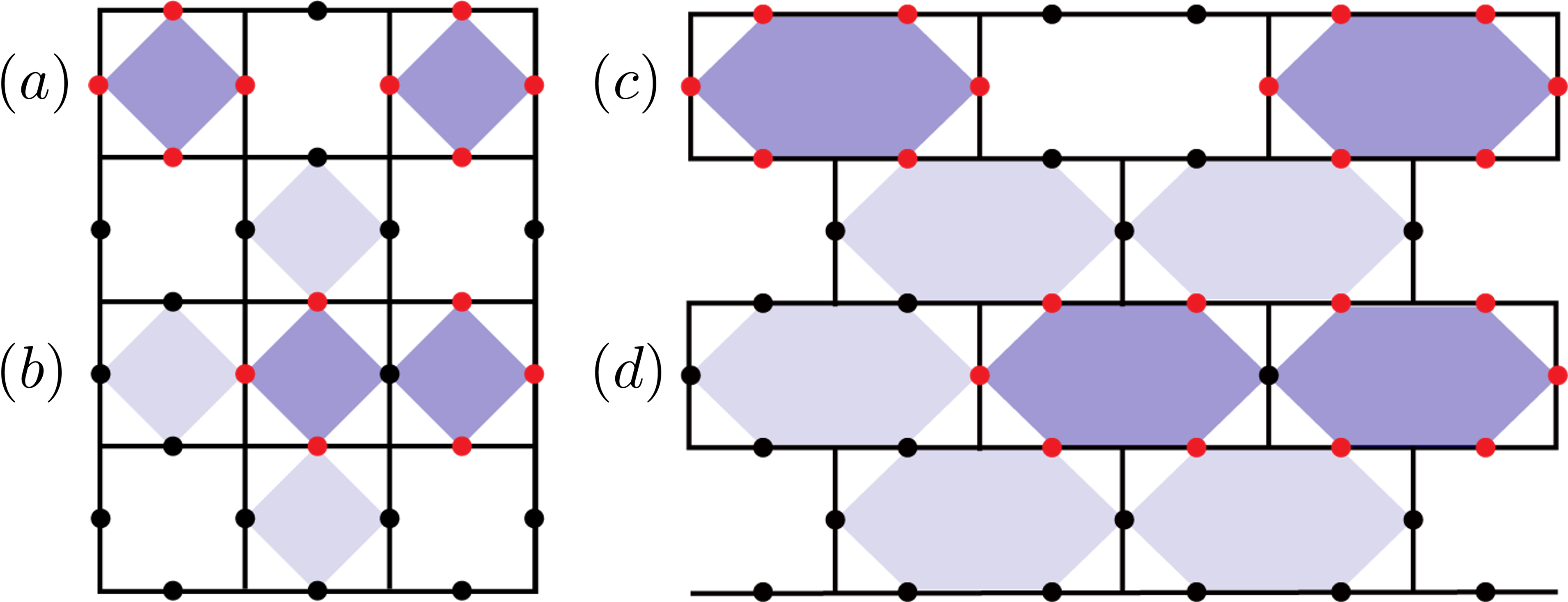}
    \caption{$F_{q,X}$ calculation for toric codes on square and honeycomb lattices. Square lattice: (a) one example of $N^2$ terms with prefactor $(1-2p)^8$; (b) one example of $N$ terms with prefactor $4(1-2p)^6$. Honeycomb lattice: (c) one example of $O(N^2)$ terms with prefactor $(1-2p)^{12}$; (d) one example of $N$ terms with prefactor $6(1-2p)^{10}$.}
    \label{fig:F_X_figure}
\end{figure}

\subsection{$F_{q,X}$ for GHZ state}

We now compute the quantum Fisher information $F_{q,X}$ for a GHZ with $N$ qubits. {\change From Eq.~\eqref{eq:QFI_is_variance}, we know that it equals four times the variance of
$O=\sum_j Z_j$.} As this quantity remains linear in the density matrix $\rho$, it can be computed by the dual channel trick
\begin{equation}
\braket{O^2}_\rho=\braket{E_X^\dagger[O^2]}_{\ket{\text{GHZ}}}.
\end{equation}
The action of the dual channel $E_X^\dagger$ on $O^2=\sum_{i,j}Z_i Z_j$ can be counted by: (1) $N^2$ terms with prefactor $(1-2p)^2$, each coming from the identity $(1-p)Z_i+pX_iZ_i X_i=(1-2p)Z_i$; (2) subtraction of the $N$ over-counted $i=j$ case that gives $Z_i^2=1$. Adding up two contributions, we have
\begin{equation}
    F_{q,X}=4(1-2p)^2N^2+4\left[1-(1-2p)^2\right]N=4(1-2p)^2N^2+16\left[p(1-p)\right]N.
\end{equation}

\subsection{$F_{q,X}$ for $\ket{\text{CSS}}$ from toric codes}

We move on to the $\ket{\text{CSS}}$ constructed from toric codes on square and hexagonal lattices. We start with the square lattice where the imprinter is chosen as $O=\sum_pS_{Z,p}$, with each plaquette containing $4$ physical qubits. As $\left\{\prod\limits_{j \in \includegraphics[width=0.2cm, height=0.15cm]{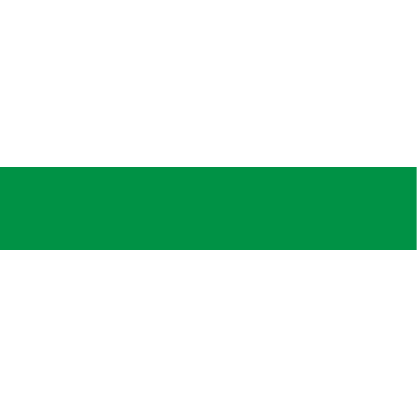}} X_j,O\right\}=0$ and $E_X$ strongly respects $\prod\limits_{j \in \includegraphics[width=0.2cm, height=0.15cm]{Figures/horizontal_line.pdf}} X_j$, the $F_{q,X}$ can be calculated in the same way as $\braket{O^2}=\sum_{p,p'}S_{Z,p}S_{Z,p'}$ using the dual-channel trick:
\begin{equation}
\braket{O^2}_\rho=\braket{E_X^\dagger[O^2]}_{\ket{\text{CSS}}}.
\end{equation}
Different from the GHZ state case, we now have three contributions in $\braket{O^2}$: (1) $N_Z^2$ (number of $S_Z$ stabilizer generators) terms with prefactor $(1-2p)^8$, coming from $8$ physical qubits on two separated $S_Z$ stabilizer generators (Fig.~\ref{fig:F_X_figure}a); (2) subtraction of $N_Z$ over-counted terms for $p=p'$ (here $p$ is the label for the plaquette, not the error strength), which give $S_{Z,p}^2=1$; (3) subtraction of $N_Z$ over-counted terms for adjacent $\braket{p,p'}$, which instead give a prefactor $4(1-2p)^6$ (Fig.~\ref{fig:F_X_figure}b) for $6$ physical qubits on $4$ different directions. Adding up three contributions, we have
\begin{equation}
    F_{q,X}=4(1-2p)^8N_Z^2+4\left[1+4(1-2p)^6-5(1-2p)^8\right]N_Z.
\end{equation}
For the toric code on the honeycomb lattice, a parallel geometric counting can be given as: (1) $N_Z^2$ terms with prefactor $(1-2p)^{12}$, coming from $12$ physical qubits on separated $S_{Z,p}$ hexagons; (2) subtraction of $N_Z$ over-counted terms for $p=p'$ and (3) subtraction of $N_Z$ over-counted terms for adjacent $\braket{p,p'}$, which give a prefactor $6(1-2p)^{10}$ for $10$ physical qubits on $6$ different directions. Adding up three terms, we have
\begin{equation}
    F_{q,X}=4(1-2p)^{12}N_Z^2+4\left[1+6(1-2p)^{10}-7(1-2p)^{12}\right]N_Z.    
\end{equation}

\section{QFI Under parallel noise and decoding for partial error correction}

\begin{figure}
    \centering
    \includegraphics[width=1\linewidth]{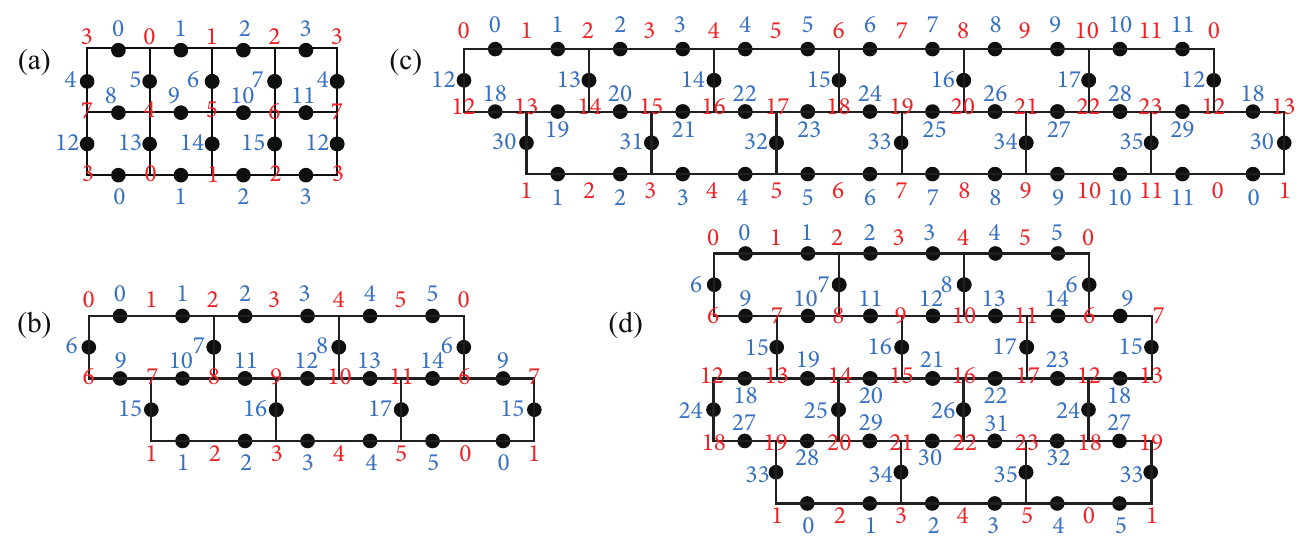}
    \caption{Decoding examples for toric codes on square (a) and honeycomb (b)-(d) lattices with different lattice sizes and geometries. Different qubits are labeled by \textcolor{myblue}{blue} numbers and vertices by \textcolor{myred}{red} numbers. For (c) and (d) we put the same number of qubits on different lattice shapes.}
    \label{fig:decoding}
\end{figure}

In this section, we provide analytical and numerical details in obtaining the $F_{q,Z}$ for various codes. We provide the details on how to conduct partial error correction for toric code construction.

Unlike the perpendicular noise, the parallel noise inevitably breaks the strong symmetry $T_X$. Specifically, for a decoherence channel whose jump operator is given by the phase imprinter,  the no-go theorem states that the asymptotic scaling in $N$ is at most linear \cite{Demkowicz-Dobrzański2012,Zhou2018}. Nevertheless, we can still find possible advantages in polynomial suppression in the noise strength when the lattice size is finite, as this is also the case where a practical sensing task is conducted. 

We begin with the GHZ state, followed by the toric and Bacon--Shor codes.

\subsection{$F_{q,Z}$ for GHZ state}

Although it is well known that the GHZ state has an exponentially decaying QFI with $N$, we find its derivation useful for later generalization. The GHZ state is a $+1$ eigenstate of the global $\mathbb{Z}_2$ symmetry $T_X=\prod_jX_j$. The i.i.d. $Z$ (dephasing) channel breaks such strong symmetry: Consider a single site dephasing $E_{Z,i}[\ket{\text{GHZ}}\bra{\text{GHZ}}]=(1-p)\ket{\text{GHZ}}\bra{\text{GHZ}}+pZ_i\ket{\text{GHZ}}\bra{\text{GHZ}} Z_i$, the dephased GHZ state has probability $P(+1)=1-p$ of being in the $+1$ eigenspace of $T_X$ and $P(-1)=p$ of being in the $-1$ eigenspace. The QFI $F_{q,Z}$ under such a channel is then calculated as \cite{W2014}
\begin{equation}
    F_{q,Z}=4\left[P(+)-P(-)\right]^2N^2.\label{eq:FZ_GHZ}
\end{equation}
We note that this formula can be explicitly derived from the definition of the QFI $2\sum_{i,j}\frac{\left(\lambda_{i}-\lambda_j\right)^2}{\lambda_i+\lambda_j}\left|\braket{\lambda_i|\sum_lZ_l|\lambda_j}\right|^2$, as the only two candidates are $\frac{1}{\sqrt{2}}\left[\ket{1\cdots1}\pm\ket{-1\cdots-1}\right]$ with probability $\lambda_{+}=1-p,\ \lambda_{-}=p$. For the i.i.d. $Z$ noise, the only modification is $P(+)-P(-)=(1-2p)^{N}$. Substituting into Eq.~\eqref{eq:FZ_GHZ}, we have 
\begin{equation}
    F_{q,Z}=4(1-2p)^{2N}N^2.
\end{equation}

\subsection{$F_{q,Z}$ for $\ket{\text{CSS}}$ from toric codes}

\begin{figure}[t]
    \centering
    \includegraphics[width=0.5\columnwidth]{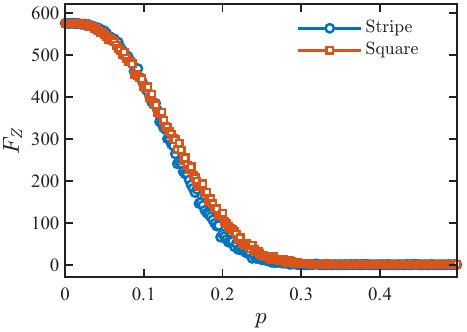}
    \caption{Quantum Fisher information $F_{q,Z}$ as a function of the parameter $p$ for two lattice configurations: stripe (Fig.~\ref{fig:decoding}c, circles, solid line) and square (Fig.~\ref{fig:decoding}d, squares, solid line). The data are obtained for a system of size $N=36$. The plot shows that different choices of finite-size lattices do not significantly change the behavior of $F_{q,Z}$.}
    \label{fig:QFI_vs_p}
\end{figure}

We now apply the above observation to the CSS probe state $\ket{\text{CSS}}=\frac{1}{\sqrt{2}}\Big(\ket{S_{Z}=1\cdots1}+\ket{S_{Z}=-1\cdots-1}\Big)$. Consider the generalized dephasing channel
\begin{equation}
    E_{S_Z}[\rho]=(1-p)\rho+pS_{Z}\rho S_{Z},
\end{equation}
the resultant is again a mixture of $\frac{1}{\sqrt{2}}\Big(\ket{S_{Z}=1\cdots1}\pm\ket{S_{Z}=-1\cdots-1}\Big)$ with probability $P(+1)=1-p$ and $P(-1)=p$. Further applying the i.i.d. $S_Z$ noise, two probabilities satisfy $P(+1)-P(-1)=(1-2p)^{N_Z}$, with $N_Z$ the number of $S_Z$ stabilizer generators. Plugging into the defining equation of QFI, we obtain 
\begin{equation}
    F_{q,Z}=4(1-2p)^{2N_Z}N_Z^2.
\end{equation}
Here, we assume that the dephasing channel is exactly of the form $E_{S_Z}$. Starting from physical i.i.d. $Z$ noise, the only change is the parity difference $P(+)-P(-)$, which depends on the decoder. For the MWPM decoder, a local dephasing event requires at least $\lceil l/2\rceil$ errors on a weight-$l$ plaquette. Thus at small $p$ the effective local error has the form $p_{\rm eff}=c p^{\lceil l/2\rceil}+\cdots$, with a geometry- and decoder-dependent coefficient $c$. The numerical results in the main text extract this leading power directly.

\begin{table*}[t]
\centering
\tiny
\setlength{\tabcolsep}{1.2pt}

\begin{minipage}{0.32\textwidth}
\centering
\begin{tabular}{c|c|c}
\hline
\textbf{Vertices} & \textbf{Adjacency} & \textbf{Path} \\
\hline
(\textcolor{myred}{0},\textcolor{myred}{1}) & \textcolor{myblue}{0,1,5,13}/\textcolor{myblue}{1,2,6,14} & \textcolor{myblue}{1} \\
(\textcolor{myred}{0},\textcolor{myred}{2}) & \textcolor{myblue}{0,1,5,13}/\textcolor{myblue}{2,3,7,15} & \textcolor{myblue}{0,3} \\
(\textcolor{myred}{0},\textcolor{myred}{3}) & \textcolor{myblue}{0,1,5,13}/\textcolor{myblue}{3,0,4,12} & \textcolor{myblue}{0} \\
(\textcolor{myred}{0},\textcolor{myred}{4}) & \textcolor{myblue}{0,1,5,13}/\textcolor{myblue}{8,9,5,13} & \textcolor{myblue}{5} \\
(\textcolor{myred}{0},\textcolor{myred}{5}) & \textcolor{myblue}{0,1,5,13}/\textcolor{myblue}{9,10,6,14} & \textcolor{myblue}{1,6} \\
(\textcolor{myred}{0},\textcolor{myred}{6}) & \textcolor{myblue}{0,1,5,13}/\textcolor{myblue}{10,11,7,15} & \textcolor{myblue}{0,3,7} \\
(\textcolor{myred}{0},\textcolor{myred}{7}) & \textcolor{myblue}{0,1,5,13}/\textcolor{myblue}{11,8,4,12} & \textcolor{myblue}{0,4} \\
\hline
(\textcolor{myred}{1},\textcolor{myred}{2}) & \textcolor{myblue}{1,2,6,14}/\textcolor{myblue}{2,3,7,15} & \textcolor{myblue}{2} \\
(\textcolor{myred}{1},\textcolor{myred}{3}) & \textcolor{myblue}{1,2,6,14}/\textcolor{myblue}{3,0,4,12} & \textcolor{myblue}{1,0} \\
(\textcolor{myred}{1},\textcolor{myred}{4}) & \textcolor{myblue}{1,2,6,14}/\textcolor{myblue}{8,9,5,13} & \textcolor{myblue}{1,5} \\
(\textcolor{myred}{1},\textcolor{myred}{5}) & \textcolor{myblue}{1,2,6,14}/\textcolor{myblue}{9,10,6,14} & \textcolor{myblue}{6} \\
(\textcolor{myred}{1},\textcolor{myred}{6}) & \textcolor{myblue}{1,2,6,14}/\textcolor{myblue}{10,11,7,15} & \textcolor{myblue}{2,7} \\
(\textcolor{myred}{1},\textcolor{myred}{7}) & \textcolor{myblue}{1,2,6,14}/\textcolor{myblue}{11,8,4,12} & \textcolor{myblue}{1,0,4} \\
\hline
(\textcolor{myred}{2},\textcolor{myred}{3}) & \textcolor{myblue}{2,3,7,15}/\textcolor{myblue}{3,0,4,12} & \textcolor{myblue}{3} \\
(\textcolor{myred}{2},\textcolor{myred}{4}) & \textcolor{myblue}{2,3,7,15}/\textcolor{myblue}{8,9,5,13} & \textcolor{myblue}{2,1,5} \\
(\textcolor{myred}{2},\textcolor{myred}{5}) & \textcolor{myblue}{2,3,7,15}/\textcolor{myblue}{9,10,6,14} & \textcolor{myblue}{2,6} \\
(\textcolor{myred}{2},\textcolor{myred}{6}) & \textcolor{myblue}{2,3,7,15}/\textcolor{myblue}{10,11,7,15} & \textcolor{myblue}{7} \\
(\textcolor{myred}{2},\textcolor{myred}{7}) & \textcolor{myblue}{2,3,7,15}/\textcolor{myblue}{11,8,4,12} & \textcolor{myblue}{3,4} \\
\hline
(\textcolor{myred}{3},\textcolor{myred}{4}) & \textcolor{myblue}{3,0,4,12}/\textcolor{myblue}{8,9,5,13} & \textcolor{myblue}{0,5} \\
(\textcolor{myred}{3},\textcolor{myred}{5}) & \textcolor{myblue}{3,0,4,12}/\textcolor{myblue}{9,10,6,14} & \textcolor{myblue}{3,2,6} \\
(\textcolor{myred}{3},\textcolor{myred}{6}) & \textcolor{myblue}{3,0,4,12}/\textcolor{myblue}{10,11,7,15} & \textcolor{myblue}{3,7} \\
(\textcolor{myred}{3},\textcolor{myred}{7}) & \textcolor{myblue}{3,0,4,12}/\textcolor{myblue}{11,8,4,12} & \textcolor{myblue}{4} \\
\hline
(\textcolor{myred}{4},\textcolor{myred}{5}) & \textcolor{myblue}{8,9,5,13}/\textcolor{myblue}{9,10,6,14} & \textcolor{myblue}{9} \\
(\textcolor{myred}{4},\textcolor{myred}{6}) & \textcolor{myblue}{8,9,5,13}/\textcolor{myblue}{10,11,7,15} & \textcolor{myblue}{8,11} \\
(\textcolor{myred}{4},\textcolor{myred}{7}) & \textcolor{myblue}{8,9,5,13}/\textcolor{myblue}{11,8,4,12} & \textcolor{myblue}{8} \\
\hline
(\textcolor{myred}{5},\textcolor{myred}{6}) & \textcolor{myblue}{9,10,6,14}/\textcolor{myblue}{10,11,7,15} & \textcolor{myblue}{10} \\
(\textcolor{myred}{5},\textcolor{myred}{7}) & \textcolor{myblue}{9,10,6,14}/\textcolor{myblue}{11,8,4,12} & \textcolor{myblue}{9,8} \\
\hline
(\textcolor{myred}{6},\textcolor{myred}{7}) & \textcolor{myblue}{10,11,7,15}/\textcolor{myblue}{11,8,4,12} & \textcolor{myblue}{11} \\
\hline
\end{tabular}
\caption{Shortest paths between vertices for the $16$-qubit toric code on the square lattice (Fig.~\ref{fig:decoding}a). The first column lists all possible pairs of vertices, whose adjacent qubits are given in the middle column. The shortest paths (sets of qubits) are given in the third column. The weight of one path is exactly the number of qubits it contains.}\label{tab: paths 16 qubits}
\end{minipage}
\hfill
\begin{minipage}{0.66\textwidth}
\centering
\begin{tabular}{c|c|c||c|c|c}
\hline
\textbf{Vertices} & \textbf{Adjacency} & \textbf{Path}
&
\textbf{Vertices} & \textbf{Adjacency} & \textbf{Path} \\
\hline
(\textcolor{myred}{0},\textcolor{myred}{1}) & \textcolor{myblue}{0,5,6}/\textcolor{myblue}{0,1,15} & \textcolor{myblue}{0} &
(\textcolor{myred}{3},\textcolor{myred}{8}) & \textcolor{myblue}{2,3,16}/\textcolor{myblue}{10,11,7} & \textcolor{myblue}{2,7} \\
(\textcolor{myred}{0},\textcolor{myred}{2}) & \textcolor{myblue}{0,5,6}/\textcolor{myblue}{1,2,7} & \textcolor{myblue}{0,1} &
(\textcolor{myred}{3},\textcolor{myred}{9}) & \textcolor{myblue}{2,3,16}/\textcolor{myblue}{11,12,16} & \textcolor{myblue}{16} \\
(\textcolor{myred}{0},\textcolor{myred}{3}) & \textcolor{myblue}{0,5,6}/\textcolor{myblue}{2,3,16} & \textcolor{myblue}{0,1,2} &
(\textcolor{myred}{3},\textcolor{myred}{10}) & \textcolor{myblue}{2,3,16}/\textcolor{myblue}{12,13,8} & \textcolor{myblue}{3,8} \\
(\textcolor{myred}{0},\textcolor{myred}{4}) & \textcolor{myblue}{0,5,6}/\textcolor{myblue}{3,4,8} & \textcolor{myblue}{5,4} &
(\textcolor{myred}{3},\textcolor{myred}{11}) & \textcolor{myblue}{2,3,16}/\textcolor{myblue}{13,14,17} & \textcolor{myblue}{3,4,17} \\
(\textcolor{myred}{0},\textcolor{myred}{5}) & \textcolor{myblue}{0,5,6}/\textcolor{myblue}{4,5,17} & \textcolor{myblue}{5} &
(\textcolor{myred}{4},\textcolor{myred}{5}) & \textcolor{myblue}{3,4,8}/\textcolor{myblue}{4,5,17} & \textcolor{myblue}{4} \\
(\textcolor{myred}{0},\textcolor{myred}{6}) & \textcolor{myblue}{0,5,6}/\textcolor{myblue}{6,9,14} & \textcolor{myblue}{6} &
(\textcolor{myred}{4},\textcolor{myred}{6}) & \textcolor{myblue}{3,4,8}/\textcolor{myblue}{6,9,14} & \textcolor{myblue}{4,5,6} \\
(\textcolor{myred}{0},\textcolor{myred}{7}) & \textcolor{myblue}{0,5,6}/\textcolor{myblue}{9,10,15} & \textcolor{myblue}{0,15} &
(\textcolor{myred}{4},\textcolor{myred}{7}) & \textcolor{myblue}{3,4,8}/\textcolor{myblue}{9,10,15} & \textcolor{myblue}{3,2,1,15} \\
(\textcolor{myred}{0},\textcolor{myred}{8}) & \textcolor{myblue}{0,5,6}/\textcolor{myblue}{10,11,7} & \textcolor{myblue}{0,1,7} &
(\textcolor{myred}{4},\textcolor{myred}{8}) & \textcolor{myblue}{3,4,8}/\textcolor{myblue}{10,11,7} & \textcolor{myblue}{3,2,7} \\
(\textcolor{myred}{0},\textcolor{myred}{9}) & \textcolor{myblue}{0,5,6}/\textcolor{myblue}{11,12,16} & \textcolor{myblue}{0,1,2,16} &
(\textcolor{myred}{4},\textcolor{myred}{9}) & \textcolor{myblue}{3,4,8}/\textcolor{myblue}{11,12,16} & \textcolor{myblue}{3,16} \\
(\textcolor{myred}{0},\textcolor{myred}{10}) & \textcolor{myblue}{0,5,6}/\textcolor{myblue}{12,13,8} & \textcolor{myblue}{5,4,8} &
(\textcolor{myred}{4},\textcolor{myred}{10}) & \textcolor{myblue}{3,4,8}/\textcolor{myblue}{12,13,8} & \textcolor{myblue}{8} \\
(\textcolor{myred}{0},\textcolor{myred}{11}) & \textcolor{myblue}{0,5,6}/\textcolor{myblue}{13,14,17} & \textcolor{myblue}{5,17} &
(\textcolor{myred}{4},\textcolor{myred}{11}) & \textcolor{myblue}{3,4,8}/\textcolor{myblue}{13,14,17} & \textcolor{myblue}{4,17} \\
\hline
(\textcolor{myred}{1},\textcolor{myred}{2}) & \textcolor{myblue}{0,1,15}/\textcolor{myblue}{1,2,7} & \textcolor{myblue}{1} &
(\textcolor{myred}{5},\textcolor{myred}{6}) & \textcolor{myblue}{4,5,17}/\textcolor{myblue}{6,9,14} & \textcolor{myblue}{5,6} \\
(\textcolor{myred}{1},\textcolor{myred}{3}) & \textcolor{myblue}{0,1,15}/\textcolor{myblue}{2,3,16} & \textcolor{myblue}{1,2} &
(\textcolor{myred}{5},\textcolor{myred}{7}) & \textcolor{myblue}{4,5,17}/\textcolor{myblue}{9,10,15} & \textcolor{myblue}{5,0,15} \\
(\textcolor{myred}{1},\textcolor{myred}{4}) & \textcolor{myblue}{0,1,15}/\textcolor{myblue}{3,4,8} & \textcolor{myblue}{0,5,4} &
(\textcolor{myred}{5},\textcolor{myred}{8}) & \textcolor{myblue}{4,5,17}/\textcolor{myblue}{10,11,7} & \textcolor{myblue}{4,3,2,7} \\
(\textcolor{myred}{1},\textcolor{myred}{5}) & \textcolor{myblue}{0,1,15}/\textcolor{myblue}{4,5,17} & \textcolor{myblue}{0,5} &
(\textcolor{myred}{5},\textcolor{myred}{9}) & \textcolor{myblue}{4,5,17}/\textcolor{myblue}{11,12,16} & \textcolor{myblue}{4,3,16} \\
(\textcolor{myred}{1},\textcolor{myred}{6}) & \textcolor{myblue}{0,1,15}/\textcolor{myblue}{6,9,14} & \textcolor{myblue}{0,6} &
(\textcolor{myred}{5},\textcolor{myred}{10}) & \textcolor{myblue}{4,5,17}/\textcolor{myblue}{12,13,8} & \textcolor{myblue}{4,8} \\
(\textcolor{myred}{1},\textcolor{myred}{7}) & \textcolor{myblue}{0,1,15}/\textcolor{myblue}{9,10,15} & \textcolor{myblue}{15} &
(\textcolor{myred}{5},\textcolor{myred}{11}) & \textcolor{myblue}{4,5,17}/\textcolor{myblue}{13,14,17} & \textcolor{myblue}{17} \\
(\textcolor{myred}{1},\textcolor{myred}{8}) & \textcolor{myblue}{0,1,15}/\textcolor{myblue}{10,11,7} & \textcolor{myblue}{1,7} &
(\textcolor{myred}{6},\textcolor{myred}{7}) & \textcolor{myblue}{6,9,14}/\textcolor{myblue}{9,10,15} & \textcolor{myblue}{9} \\
(\textcolor{myred}{1},\textcolor{myred}{9}) & \textcolor{myblue}{0,1,15}/\textcolor{myblue}{11,12,16} & \textcolor{myblue}{1,2,16} &
(\textcolor{myred}{6},\textcolor{myred}{8}) & \textcolor{myblue}{6,9,14}/\textcolor{myblue}{10,11,7} & \textcolor{myblue}{9,10} \\
(\textcolor{myred}{1},\textcolor{myred}{10}) & \textcolor{myblue}{0,1,15}/\textcolor{myblue}{12,13,8} & \textcolor{myblue}{0,5,4,8} &
(\textcolor{myred}{6},\textcolor{myred}{9}) & \textcolor{myblue}{6,9,14}/\textcolor{myblue}{11,12,16} & \textcolor{myblue}{9,10,11} \\
(\textcolor{myred}{1},\textcolor{myred}{11}) & \textcolor{myblue}{0,1,15}/\textcolor{myblue}{13,14,17} & \textcolor{myblue}{0,5,17} &
(\textcolor{myred}{6},\textcolor{myred}{10}) & \textcolor{myblue}{6,9,14}/\textcolor{myblue}{12,13,8} & \textcolor{myblue}{14,13} \\
\hline
(\textcolor{myred}{2},\textcolor{myred}{3}) & \textcolor{myblue}{1,2,7}/\textcolor{myblue}{2,3,16} & \textcolor{myblue}{2} &
(\textcolor{myred}{6},\textcolor{myred}{11}) & \textcolor{myblue}{6,9,14}/\textcolor{myblue}{13,14,17} & \textcolor{myblue}{14} \\
(\textcolor{myred}{2},\textcolor{myred}{4}) & \textcolor{myblue}{1,2,7}/\textcolor{myblue}{3,4,8} & \textcolor{myblue}{2,3} &
(\textcolor{myred}{7},\textcolor{myred}{8}) & \textcolor{myblue}{9,10,15}/\textcolor{myblue}{10,11,7} & \textcolor{myblue}{10} \\
(\textcolor{myred}{2},\textcolor{myred}{5}) & \textcolor{myblue}{1,2,7}/\textcolor{myblue}{4,5,17} & \textcolor{myblue}{1,0,5} &
(\textcolor{myred}{7},\textcolor{myred}{9}) & \textcolor{myblue}{9,10,15}/\textcolor{myblue}{11,12,16} & \textcolor{myblue}{10,11} \\
(\textcolor{myred}{2},\textcolor{myred}{6}) & \textcolor{myblue}{1,2,7}/\textcolor{myblue}{6,9,14} & \textcolor{myblue}{1,0,6} &
(\textcolor{myred}{7},\textcolor{myred}{10}) & \textcolor{myblue}{9,10,15}/\textcolor{myblue}{12,13,8} & \textcolor{myblue}{9,14,13} \\
(\textcolor{myred}{2},\textcolor{myred}{7}) & \textcolor{myblue}{1,2,7}/\textcolor{myblue}{9,10,15} & \textcolor{myblue}{1,15} &
(\textcolor{myred}{7},\textcolor{myred}{11}) & \textcolor{myblue}{9,10,15}/\textcolor{myblue}{13,14,17} & \textcolor{myblue}{9,14} \\
(\textcolor{myred}{2},\textcolor{myred}{8}) & \textcolor{myblue}{1,2,7}/\textcolor{myblue}{10,11,7} & \textcolor{myblue}{7} &
(\textcolor{myred}{8},\textcolor{myred}{9}) & \textcolor{myblue}{10,11,7}/\textcolor{myblue}{11,12,16} & \textcolor{myblue}{11} \\
(\textcolor{myred}{2},\textcolor{myred}{9}) & \textcolor{myblue}{1,2,7}/\textcolor{myblue}{11,12,16} & \textcolor{myblue}{2,16} &
(\textcolor{myred}{8},\textcolor{myred}{10}) & \textcolor{myblue}{10,11,7}/\textcolor{myblue}{12,13,8} & \textcolor{myblue}{11,12} \\
(\textcolor{myred}{2},\textcolor{myred}{10}) & \textcolor{myblue}{1,2,7}/\textcolor{myblue}{12,13,8} & \textcolor{myblue}{2,3,8} &
(\textcolor{myred}{8},\textcolor{myred}{11}) & \textcolor{myblue}{10,11,7}/\textcolor{myblue}{13,14,17} & \textcolor{myblue}{10,9,14} \\
(\textcolor{myred}{2},\textcolor{myred}{11}) & \textcolor{myblue}{1,2,7}/\textcolor{myblue}{13,14,17} & \textcolor{myblue}{1,0,5,17} &
(\textcolor{myred}{9},\textcolor{myred}{10}) & \textcolor{myblue}{11,12,16}/\textcolor{myblue}{12,13,8} & \textcolor{myblue}{12} \\
(\textcolor{myred}{3},\textcolor{myred}{4}) & \textcolor{myblue}{2,3,16}/\textcolor{myblue}{3,4,8} & \textcolor{myblue}{3} &
(\textcolor{myred}{9},\textcolor{myred}{11}) & \textcolor{myblue}{11,12,16}/\textcolor{myblue}{13,14,17} & \textcolor{myblue}{12,13} \\
(\textcolor{myred}{3},\textcolor{myred}{5}) & \textcolor{myblue}{2,3,16}/\textcolor{myblue}{4,5,17} & \textcolor{myblue}{3,4} &
(\textcolor{myred}{10},\textcolor{myred}{11}) & \textcolor{myblue}{12,13,8}/\textcolor{myblue}{13,14,17} & \textcolor{myblue}{13} \\
(\textcolor{myred}{3},\textcolor{myred}{6}) & \textcolor{myblue}{2,3,16}/\textcolor{myblue}{6,9,14} & \textcolor{myblue}{2,1,0,6} & & & \\
(\textcolor{myred}{3},\textcolor{myred}{7}) & \textcolor{myblue}{2,3,16}/\textcolor{myblue}{9,10,15} & \textcolor{myblue}{2,1,15} & & & \\
\hline
\end{tabular}
\caption{Shortest paths between vertices for the $18$-qubit toric code on the honeycomb lattice (Fig.~\ref{fig:decoding}b). The first column lists all possible pairs of vertices, whose adjacent qubits are given in the middle column. The shortest paths are given in the third column. The weight of one path is exactly the number of qubits it contains.}\label{tab: paths 18 qubits}
\end{minipage}

\end{table*}

To see how the decoding and error correction are realized, we next provide several sample codes for different background lattices. We start with a simple $16$-qubit toric code on the square lattice in Fig.~\ref{fig:decoding}a. We label qubits by \textcolor{myblue}{blue} numbers and vertices by \textcolor{myred}{red} numbers. We also label the syndrome outcome by its vertex numbers. For example, when a $Z$ error string connects $\textcolor{myred}{0}$ and $\textcolor{myred}{2}$, we have the syndrome outcome $S_{\textcolor{myred}{0,2}}=-1$ while the others are $1$. To apply the Minimum-Weight Perfect Matching (MWPM) decoder, we next assign minimum weights to all paths between vertices. As we are applying uniform errors to all physical qubits, the minimum weight is chosen as the shortest distance between two vertices. For the $16$-qubit code, such shortest paths are given in Tab.~\ref{tab: paths 16 qubits}, where the first column lists all possible pairs of vertices, the middle column lists the corresponding adjacent qubits for two vertices, and the third column computes the shortest path between them. The weight of a path is exactly the number of qubits it contains. For example, $\textcolor{myred}{(0,1)}$ has weight $1$ and $\textcolor{myred}{(0,6)}$ has weight $3$. After a round of syndrome measurement, we have an even number of $-1$ syndromes. We then compute the minimum weight for all pairs of syndromes/vertices, and apply $Z's$ along the overall shortest paths. A parallel weight computation for the $18$-qubit toric code on the honeycomb lattice (Fig.~\ref{fig:decoding}b) can be done with a different (shortest) path identification in Tab.~\ref{tab: paths 18 qubits}.

{\change We condition the toric-code numerics on samples for which no logical \(Z_L\) error remains after recovery. This isolates the local dephasing and is not intended as a heralded experimental postselection protocol. For code families operated below threshold with growing code distance, the neglected logical error is expected to be exponentially suppressed in the code distance.} We therefore expect the QFI $F_{q,Z}$ to be determined by the bulk properties of the code and keep only patterns with no logical errors after error correction. For example, for the stripe lattice in Fig.~\ref{fig:decoding}a, we keep only patterns satisfying $X_4X_5X_6X_7=1$ and $X_0X_8=1$. Similarly, we require $X_6X_7X_8=1$ and $X_1X_{10}=1$ for (b), $X_{12}X_{13}X_{14}X_{15}X_{16}X_{17}=1$ and $X_1X_{19}=1$ for (c), and $X_6X_7X_8=1$ and $X_1X_{10}X_{19}X_{28}=1$ for (d). After this filtering, we expect that the QFI $F_{q,Z}$ is solely determined by the bulk property of the code rather than lattice shapes. Indeed, in Fig.~\ref{fig:QFI_vs_p}, we observe only minor misalignment for stripe lattice in Fig.~\ref{fig:decoding}c and square lattice in Fig.~\ref{fig:decoding}d with identical number of qubits $N=36$.

With this assumption, the QFI $F_{q,Z}$ can be explicitly computed using Eq.~\eqref{eq:FZ_GHZ}. {\change Let $n_{\rm tot}=n_++n_-$ denote the number
of retained samples and define
$P(\pm)=n_\pm/n_{\rm tot}$.
} $n_\pm$ here is, after error correction, the number of results with positive and negative parity: $X_0X_1X_2X_3=\pm1$ (enforced symmetry in the probe state) for Fig.~\ref{fig:decoding}a, $X_0X_1X_2X_3X_4X_5=\pm1$ for Fig.~\ref{fig:decoding}b, $X_0X_1X_2X_3\cdots X_{11}=\pm1$ for Fig.~\ref{fig:decoding}c, and $\left[X_0X_1X_2X_3X_4X_5\right]\left[X_{18}X_{19}X_{20}X_{21}X_{22}X_{23}\right]=\pm1$ for Fig.~\ref{fig:decoding}d. This classical method can be directly applied to larger lattice sizes to compute the QFI $F_{q,Z}$.

\subsection{$F_{q,Z}$ for $\ket{\text{CSS}}$ from Bacon-Shor codes}

For Bacon--Shor, the relevant decoder is majority voting within a row of odd length $n$. Define the row-failure probability
\begin{equation}
q_n(p)=\sum_{k=(n+1)/2}^{n}\binom{n}{k}p^k(1-p)^{n-k}.
\end{equation}
Only a failure on an odd row changes the $T_X$ parity connecting the two branches, and Eq.~\eqref{eq:FZ_GHZ} gives
\begin{equation}
F_{q,Z}=4(m-1)^2[1-2q_n(p)]^{m}.
\end{equation}
For large $n$, $q_n(p)=\mathcal{O}(n^{-1/2}[4p(1-p)]^{n/2})$, which leads to the Bacon--Shor-specific threshold discussed in the main text. We also refer interested readers to \cite{W2014} for further details.

\section{Proof of Theorem~1:~QFI under simultaneous parallel and
perpendicular noise for partial error correction}

\begin{figure}[b]
    \centering
    \includegraphics[width=0.9\linewidth]{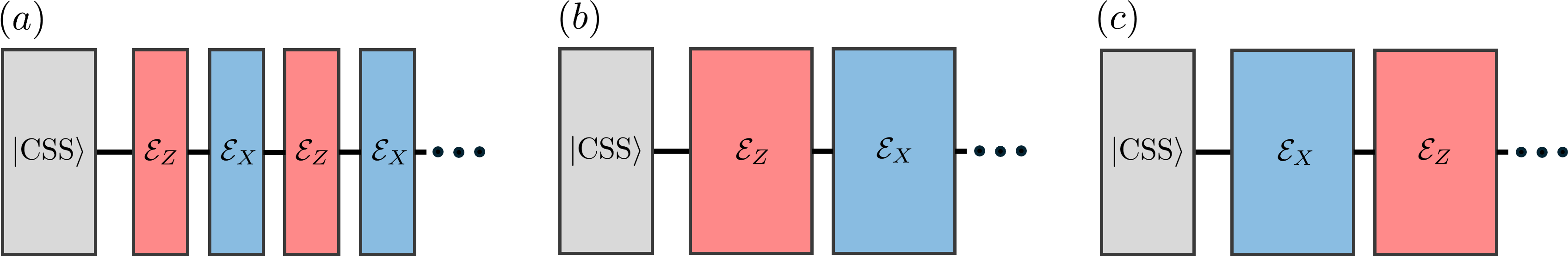}
    \caption{
Composition of independent i.i.d.\ $Z$ and $X$ noise channels.
The alternating application of $Z$ and $X$ noise in (a) is
equivalent to applying the total $Z$-noise channel followed by
the total $X$-noise channel in (b), or in the reverse order in (c).
}
    \label{fig:mixture_noise}
\end{figure}

In this appendix, we prove a major argument in the main text:
the super-SQL scaling persists under simultaneous
parallel and perpendicular noise. We first consider a GHZ state
under simultaneous i.i.d.\ $Z$ and $X$ noise channels, and then
extend the calculation to the $\ket{\mathrm{CSS}}$ probe state.

We focus here on errors occurring before phase imprinting and
postpone the case of errors after phase imprinting to the next
appendix. We consider i.i.d.\ $Z$ and $X$ noise
channels, $\mathcal{E}_Z$ and $\mathcal{E}_X$, applied in an
alternating manner, as illustrated in Fig.~\ref{fig:mixture_noise}a.
Since the two Pauli channels commute,
\[
\mathcal{E}_Z\circ\mathcal{E}_X
=
\mathcal{E}_X\circ\mathcal{E}_Z,
\]
all $Z$- and $X$-noise steps can be collected into total channels
with effective strengths $p_Z$ and $p_X$. Without loss of
generality, we therefore apply the total i.i.d.\ $X$ noise first.

Assume the total i.i.d. $X$ noise takes the following form $\mathcal{E}_X=\prod_j\mathcal{E}_{X,j}$ with 
\begin{equation}
    \mathcal{E}_{X,j}[\rho]=(1-p_X)\rho+p_X X_j\rho X_j,
\end{equation}
the GHZ state after $\mathcal{E}_X$ becomes a mixture of  $\ket{s,\Bar{s},+}\equiv\frac{1}{\sqrt{2}}\left[\ket{s}+\ket{\Bar{s}}\right]$, with $s=(s_1,s_2,...,s_N)\in\left\{1,-1\right\}^{N}$ a binary bit string recording the $s_Z$ for each qubit. We introduce by $\Bar{s}$ the bitwise complement of $s$. Throughout this section, each complement pair $\{s,\bar{s}\}$
is counted only once, and all sums over $s$ are understood to run
over one representative from each pair. Specifically, we have
\begin{equation}
    \mathcal{E}_{X}\left[\ket{\text{GHZ}}\bra{\text{GHZ}}\right]=\sum_{\left\{s,\Bar{s}\right\}}p_{\left\{s,\Bar{s}\right\}}\ket{s,\Bar{s},+}\bra{s,\Bar{s},+}.
    \label{eq:rho_X_GHZ}
\end{equation}
Importantly, we have the orthogonality condition $\braket{s_1,\Bar{s}_1,+|s_2,\Bar{s}_2,+}=\delta_{s_1,s_2}$. The diagonal element $p_{\left\{s,\Bar{s}\right\}}=p_X^{h(s)}(1-p_{X})^{N-h(s)}+p_{X}^{N-h(s)}(1-p_X)^{h(s)}$ with
\[
h(s)=\sum_j\frac{1-s_j}{2}
\]
the Hamming weight counting the number of $-1$ entries in $s$. Plugging Eq.~\eqref{eq:rho_X_GHZ} into the defining equation of the QFI, and note that $\sum_jZ_j\ket{s,\Bar{s},+}=\left[N-2h(s)\right]\ket{s,\Bar{s},-}$ with $\ket{s,\Bar{s},-}\equiv\frac{1}{\sqrt{2}}\left[\ket{s}-\ket{\Bar{s}}\right]$, we have
\begin{align}
    F_{q,X}=&2\sum_{\left\{s,\Bar{s}\right\}}\frac{\left(p_{\left\{s,\Bar{s}\right\}}-0\right)^2}{p_{\left\{s,\Bar{s}\right\}}+0}\left|\braket{s,\Bar{s},+|\sum_jZ_j|s,\Bar{s},-}\right|^2+2\sum_{\left\{s,\Bar{s}\right\}}\frac{\left(0-p_{\left\{s,\Bar{s}\right\}}\right)^2}{0+p_{\left\{s,\Bar{s}\right\}}}\left|\braket{s,\Bar{s},-|\sum_jZ_j|s,\Bar{s},+}\right|^2\\
    =&4\sum_{\left\{s,\Bar{s}\right\}}p_{\left\{s,\Bar{s}\right\}}\left[N-2h(s)\right]^2 \\
    =&4\sum_{m=0}^{N}\binom{N}{m}p_X^m(1-p_X)^{N-m}(N-2m)^2\\
    =&4(1-2p_X)^2N^2+4\left[1-(1-2p_X)^2\right]N,\label{eq:FX_GHZ}
\end{align}
where we have used the orthogonality $\braket{s,\Bar{s},+|s,\Bar{s},-}=0$. We note that the last line is exactly $4$ times the variance of $\sum_jZ_j$ with respect to $\mathcal{E}_X\left[\ket{\text{GHZ}}\bra{\text{GHZ}}\right]$, which is derived conveniently from the symmetry argument in earlier context. Here, we write down another method to compute it, which is helpful when the variance argument breaks down for a broken strong symmetry.

We next apply the i.i.d. $Z$ noise $\mathcal{E}_{Z}$ to the GHZ state. At this stage, no error correction is possible and a pure $\ket{s,\Bar{s},+}$ becomes 
\begin{equation}
    \mathcal{E}_{Z}\left[\ket{s,\Bar{s},+}\bra{s,\Bar{s},+}\right]=p_+\ket{s,\Bar{s},+}\bra{s,\Bar{s},+}+p_-\ket{s,\Bar{s},-}\bra{s,\Bar{s},-},
\end{equation}
with $p_\pm=\frac{1\pm(1-2p_Z)^N}{2}$ \textit{independent} of $s,\Bar{s}$. Therefore, the GHZ state after \textbf{both} $\mathcal{E}_X$ and $\mathcal{E}_Z$ can be diagonalized as
\begin{equation}
    \mathcal{E}_{Z}\cdot\mathcal{E}_X\left[\ket{\text{GHZ}}\bra{\text{GHZ}}\right]=\sum_{\left\{s,\Bar{s}\right\}}p_{\left\{s,\Bar{s}\right\}}\left(p_+\ket{s,\Bar{s},+}\bra{s,\Bar{s},+}+p_-\ket{s,\Bar{s},-}\bra{s,\Bar{s},-}\right).\label{eq:rho_XZ_GHZ}
\end{equation}
Importantly, $\ket{s,\Bar{s},\pm}$ form an orthogonal basis and, therefore, the QFI can be computed as
\begin{align}
    F_{q,XZ}=&2\sum_{\{s_1,\bar{s}_1\},\{s_2,\bar{s}_2\}}.\frac{\left(p_{\left\{s_1,\Bar{s}_1\right\}}p_+-p_{\left\{s_2,\Bar{s}_2\right\}}p_-\right)^2}{p_{\left\{s_1,\Bar{s}_1\right\}}p_++p_{\left\{s_2,\Bar{s}_2\right\}}p_-}\left|\braket{s_2,\Bar{s}_2,+|\sum_jZ_j|s_1,\Bar{s}_1,-}\right|^2\\
    &+2\sum_{\{s_1,\bar{s}_1\},\{s_2,\bar{s}_2\}}.\frac{\left(p_{\left\{s_1,\Bar{s}_1\right\}}p_--p_{\left\{s_2,\Bar{s}_2\right\}}p_+\right)^2}{p_{\left\{s_1,\Bar{s}_1\right\}}p_-+p_{\left\{s_2,\Bar{s}_2\right\}}p_+}\left|\braket{s_2,\Bar{s}_2,-|\sum_jZ_j|s_1,\Bar{s}_1,+}\right|^2\\
    =&2\sum_{\{s_1,\bar{s}_1\},\{s_2,\bar{s}_2\}}.\frac{\left(p_{\left\{s_1,\Bar{s}_1\right\}}p_+-p_{\left\{s_2,\Bar{s}_2\right\}}p_-\right)^2}{p_{\left\{s_1,\Bar{s}_1\right\}}p_++p_{\left\{s_2,\Bar{s}_2\right\}}p_-}\left[N-2h(s_1)\right]^2\delta_{s_1,s_2}\\
    &+2\sum_{\{s_1,\bar{s}_1\},\{s_2,\bar{s}_2\}}.\frac{\left(p_{\left\{s_1,\Bar{s}_1\right\}}p_--p_{\left\{s_2,\Bar{s}_2\right\}}p_+\right)^2}{p_{\left\{s_1,\Bar{s}_1\right\}}p_-+p_{\left\{s_2,\Bar{s}_2\right\}}p_+}\left[N-2h(s_1)\right]^2\delta_{s_1,s_2}\\
    =&4\left(p_+-p_-\right)^2\sum_{\left\{s,\Bar{s}\right\}}p_{\left\{s,\Bar{s}\right\}}\left[N-2h(s)\right]^2\\
    =&4(1-2p_Z)^{2N}\left\{(1-2p_X)^2N^2+\left[1-(1-2p_X)^2\right]N\right\}\\
    =&\frac{F_{q,X}F_{q,Z}}{4N^2},
\end{align}
where $F_{q,Z}=4(1-2p_Z)^{2N}N^2$ is the QFI under \textit{sole} $\mathcal{E}_Z$ noise and $F_{q,X}$ is the QFI under \textit{sole} $\mathcal{E}_X$ noise given in Eq.~\eqref{eq:FX_GHZ} for the GHZ state. We can easily see that the GHZ state itself is fragile to
simultaneous $X$ and $Z$ noise even if its $F_{q,X}$ preserves the HL scaling, because $F_{q,Z}$
has sub-SQL scaling.

The same computation for $F_{q,XZ}$ can be done similarly as for the $\ket{\text{CSS}}$ state, by the following modification: 

1. the bit string $s=(s_1,...,s_{N_{Z}})$ now denotes the value of $S_{Z}$ syndrome and the relevant $\ket{s,\Bar{s},\pm}$ are defined as 
\begin{equation}
    \ket{s,\Bar{s},\pm}=\frac{1}{\sqrt{2}}\left(\ket{S_{Z,1}=s1,...S_{Z,N_Z}=s_{N_Z}}\pm\ket{S_{Z,1}=\Bar{s}1,...S_{Z,N_Z}=\Bar{s}_{N_Z}}\right).
\end{equation}
Any additional logical-sector labels are suppressed for notational
simplicity and are chosen consistently for the two complementary
branches. Their orthogonality condition remains as $\braket{s_{1},\Bar{s}_1,\pm|s_2,\Bar{s}_2,\pm}=\delta_{s_1,s_2}$ and $\braket{s_{1},\Bar{s}_1,\mp|s_2,\Bar{s}_2,\pm}=0$. 

2. The expression for $p_{s,\Bar{s}}$ is modified. The corresponding $F_{q,X}$ is now directly computed as four times the
variance of $\sum_\beta S_{Z,\beta}$, since the strong $T_X$ symmetry is preserved. However, structurally, we still write $\mathcal{E}_X[\ket{\text{CSS}}\bra{\text{CSS}}]=\sum_{s,\Bar{s}}p_{\left\{s,\Bar{s}\right\}}\ket{s,\Bar{s},+}\bra{s,\Bar{s},+}$ and
\begin{equation}
    F_{q,X}=4\sum_{\left\{s,\Bar{s}\right\}}p_{\left\{s,\Bar{s}\right\}}\left[N_Z-2h(s)\right]^2.
\end{equation}

3. The expression for $p_\pm$ is also modified. Nevertheless, under the assumption that no logical $Z$ error occurs, $p_\pm$ remains independent of $s,\Bar{s}$ and $F_{q,Z}$ remains as Eq.~\eqref{eq:FZ_GHZ}.

With the above clarification, we can write down $F_{q,XZ}$ in exactly the same structure as the GHZ state case, giving
\begin{align}
    F_{q,XZ}=&2\sum_{\{s_1,\bar{s}_1\},\{s_2,\bar{s}_2\}}.\frac{\left(p_{\left\{s_1,\Bar{s}_1\right\}}p_+-p_{\left\{s_2,\Bar{s}_2\right\}}p_-\right)^2}{p_{\left\{s_1,\Bar{s}_1\right\}}p_++p_{\left\{s_2,\Bar{s}_2\right\}}p_-}\left|\braket{s_2,\Bar{s}_2,+|\sum_\beta S_{Z,\beta}|s_1,\Bar{s}_1,-}\right|^2\\
    &+2\sum_{\{s_1,\bar{s}_1\},\{s_2,\bar{s}_2\}}.\frac{\left(p_{\left\{s_1,\Bar{s}_1\right\}}p_--p_{\left\{s_2,\Bar{s}_2\right\}}p_+\right)^2}{p_{\left\{s_1,\Bar{s}_1\right\}}p_-+p_{\left\{s_2,\Bar{s}_2\right\}}p_+}\left|\braket{s_2,\Bar{s}_2,-|\sum_\beta S_{Z,\beta}|s_1,\Bar{s}_1,+}\right|^2\\
    =&2\sum_{\{s_1,\bar{s}_1\},\{s_2,\bar{s}_2\}}.\frac{\left(p_{\left\{s_1,\Bar{s}_1\right\}}p_+-p_{\left\{s_2,\Bar{s}_2\right\}}p_-\right)^2}{p_{\left\{s_1,\Bar{s}_1\right\}}p_++p_{\left\{s_2,\Bar{s}_2\right\}}p_-}\left[N_Z-2h(s_1)\right]^2\delta_{s_1,s_2}\\
    &+2\sum_{\{s_1,\bar{s}_1\},\{s_2,\bar{s}_2\}}.\frac{\left(p_{\left\{s_1,\Bar{s}_1\right\}}p_--p_{\left\{s_2,\Bar{s}_2\right\}}p_+\right)^2}{p_{\left\{s_1,\Bar{s}_1\right\}}p_-+p_{\left\{s_2,\Bar{s}_2\right\}}p_+}\left[N_Z-2h(s_1)\right]^2\delta_{s_1,s_2}\\
    =&4\left(p_+-p_-\right)^2\sum_{\left\{s,\Bar{s}\right\}}p_{\left\{s,\Bar{s}\right\}}\left[N_Z-2h(s)\right]^2\\
    =&\frac{F_{q,X}F_{q,Z}}{4N_Z^2}.\label{eq:F_XZ CSS}
\end{align}
Equation~\eqref{eq:F_XZ CSS} shows that the QFI under simultaneous
parallel and perpendicular noise is given by the normalized product
of the two individual QFIs. As shown in Thm.~1, a super-SQL scaling in $F_{q,XZ}$ can be
satisfied below a finite noise threshold by appropriately increasing
the stabilizer generator weight with system size. 


{\change\section{Proof of Theorem~2}

\noindent\textbf{Theorem~2}.~\emph{
At the operating point
$\theta_\star=\pi/(4N_Z)\ {\rm mod}\ \pi/(2N_Z)$,
the QFI of the probe state $\rho_\theta$ under i.i.d.\ Pauli errors
described by
$E_P[\rho_\theta]=(1-p_X-p_Y-p_Z)\rho_\theta
+p_XX\rho_\theta X+p_YY\rho_\theta Y+p_ZZ\rho_\theta Z$
followed by partial error correction is lower bounded by the QFI
under the i.i.d.\ dephasing channel
$\mathcal{E}_{z}[\rho_\theta]=(1-p)\rho_\theta+pZ\rho_\theta Z$
followed by partial error correction, with $p=p_Y+p_Z$.}

\begin{proof}
We choose the detection operator as $T_X$ and evaluate $\braket{T_X}_\theta$ under $E_P$ followed by error correction using $S_X$. We first note that it has the same $\braket{T_X}_\theta$ as the i.i.d.\ channel
\[
    \mathcal{E}_j(\rho_\theta) = (1-p_X-p_Y-p_Z)\rho_\theta+p_XX_j\rho_\theta X_j +(p_Y+p_Z)Z_j\rho_\theta Z_j
\]
followed by the same error correction, since either $\{Z_j,T_X\}=\{Y_j,T_X\}=0$ or $[Z_j,T_X]=[Y_j,T_X]=0$, and the $S_X$ stabilizers cannot distinguish the corresponding $Y_j$ and $Z_j$ errors. As $[X_j,T_X]=0$, the $X_j$ component has identity action on $\braket{T_X}_\theta$. Therefore the same expectation value is obtained from the i.i.d.\ dephasing channel
\[
    \mathcal{E}_{z}[\rho_\theta]=(1-p_Y-p_Z)\rho_\theta +(p_Y+p_Z)Z\rho_\theta Z
\]
followed by partial error correction.

For this effective dephasing channel,
\[
    \braket{T_X}_\theta =(1-2p_{\rm odd})\cos(2N_Z\theta),
\]
where $p_{\rm odd}$ is the probability that an odd number of $S_Z$ phase flips remain after correcting the $S_X$ syndromes. Since $T_X^2=I$, measuring $T_X$ gives the classical Fisher information
\begin{equation}
    F_C(\theta) =
        \frac{
        4N_Z^2(1-2p_{\rm odd})^2
        \sin^2(2N_Z\theta)}{
        1-(1-2p_{\rm odd})^2
        \cos^2(2N_Z\theta)}.
\end{equation}
At $\theta_\star=\pi/(4N_Z)\ {\rm mod}\ \pi/(2N_Z)$, this becomes
\begin{equation}
    F_C(\theta_\star) = 4N_Z^2(1-2p_{\rm odd})^2.
\end{equation}
Because the QFI is no smaller than the Fisher information of any particular measurement,
\[
    F_q^{P}(\theta_\star) \geq 4N_Z^2(1-2p_{\rm odd})^2.
\]
The right-hand side is exactly the QFI $F_{q,Z}$ of the corresponding dephasing-only channel followed by partial error correction. Therefore,
\[
    F_q^{P}(\theta_\star)\geq F_{q,Z},
\]
which also implies $\max_\theta F_q^{P}(\theta)\geq F_{q,Z}$.
\end{proof}}

\end{document}